\documentclass[a4paper,UKenglish,cleveref, autoref, thm-restate]{lipics-v2021}
\usepackage{graphicx} 
\usepackage{graphicx} \usepackage{xspace} 
\usepackage[numbers]{natbib} 
\usepackage{hyperref,color}
\usepackage{todonotes}
\usepackage{pifont}
\usepackage{appendix}
\usepackage{amsthm}
\newcommand{\cpp}{\textsc{Cycle Packing}}
\newcommand{\cp}{\textsc{Cycle Packing }}
\newcommand{\fvs}{\textsc{Feedback Vertex Set}}

\newcommand{\vc}{\textsc{Vertex Cover}}

\newcommand{\hgc}{$\boxplus_\irr$}
\newcommand{\hg}{$\boxplus_\irr$ }
\newcommand{\irr}{\textsf{core}}

\newcommand{\ud}{\textsf{UD}}

\newcommand{\cut}[1]{\textsf{cut}({#1})}

\newcommand{\wsqr}[1]{\ensuremath{c_{\textsf{sos}}({#1}})}

\newcommand{\wmax}[1]{\ensuremath{c_\textsf{max}}(#1)}

\newcommand{\tc}{\textsf{tc}}

\newcommand{\ccheck}[1]{\textcolor{red}{#1}}

\pdfoutput=1 
\hideLIPIcs  

\title{Dynamic Parameterized Problems on Unit Disk Graphs} 

\author{Shinwoo An}{Department of Computer Science and Engineering, POSTECH, Pohang, Korea}{shinwooan@postech.ac.kr}{}{}
\author{Kyungjin Cho}{Department of Computer Science and Engineering, POSTECH, Pohang, Korea}{kyungjincho@postech.ac.kr}{0000-0003-2223-4273}{Supported by the National Research Foundation of Korea (NRF) grant funded by the Korea government (MSIT) (No.RS-2024-00410835).} 
\author{Leo Jang}{Department of Computer Science and Engineering, POSTECH, Pohang, Korea}{leo630@postech.ac.kr}{}{} 
\author{Byeonghyeon Jung}{Department of Computer Science and Engineering, POSTECH, Pohang, Korea}{bhjung@postech.ac.kr}{}{} 
\author{Yudam Lee}{Department of Computer Science and Engineering, POSTECH, Pohang, Korea}{leeyudam@postech.ac.kr}{}{} 
\author{Eunjin Oh}{Department of Computer Science and Engineering, POSTECH, Pohang, Korea}{eunjin.oh@postech.ac.kr}{0000-0003-0798-2580}{}
\author{Donghun Shin}{Department of Computer Science and Engineering, POSTECH, Pohang, Korea}{sdh728@postech.ac.kr}{}{}
\author{Hyeonjun Shin}{Department of Computer Science and Engineering, POSTECH, Pohang, Korea}{hyeonjun.shin@postech.ac.kr}{0009-0008-4701-7295}{} 
\author{Chanho Song}{Department of Computer Science and Engineering, POSTECH, Pohang, Korea}{sch0622@postech.ac.kr}{0009-0001-3522-3517}{}

\authorrunning{S An, et al.} 
\Copyright{Shinwoo An, Kyungjin Cho, Leo Jang, Byeonghyeon Jung, Yudam Lee, Eunjin Oh, Donghun Shin, Hyeonjun Shin, and Chanho Song} 

\ccsdesc[500]{Theory of computation~Computational geometry}

\keywords{Unit disk graphs, dynamic parameterized algorithms, kernelization} 

\category{} 

\funding{This work was partly supported by Institute of Information \& Communications Technology Planning \& Evaluation (IITP) grant funded by the Korea government (MSIT) (No.RS-2024-00440239, Sublinear Scalable Algorithms for Large-Scale Data Analysis) and the National Research Foundation of Korea (NRF) grant funded by the Korea government (MSIT) (No.RS-2024-00358505).}

\nolinenumbers 

\EventEditors{John Q. Open and Joan R. Access}
\EventNoEds{2}
\EventLongTitle{42nd Conference on Very Important Topics (CVIT 2016)}
\EventShortTitle{CVIT 2016}
\EventAcronym{CVIT}
\EventYear{2016}
\EventDate{December 24--27, 2016}
\EventLocation{Little Whinging, United Kingdom}
\EventLogo{}
\SeriesVolume{42}
\ArticleNo{23}

\begin{document}
\maketitle

\begin{abstract}
In this paper, we study fundamental parameterized problems such as 
 \textsc{$k$-Path/Cycle}, \textsc{Vertex Cover}, \textsc{Triangle Hitting Set}, \textsc{Feedback Vertex Set}, and \textsc{Cycle Packing} 
 for \emph{dynamic} unit disk graphs. Given a vertex set $V$ changing dynamically under vertex insertions and deletions, our goal is to maintain data structures so that
 the aforementioned parameterized problems on the unit disk graph induced by $V$ can be solved
 efficiently.
Although dynamic parameterized problems on general graphs have been studied extensively, 
no previous work focuses on unit disk graphs. In this paper, we present the first data structures
for fundamental parameterized problems on dynamic unit disk graphs.
More specifically, our data structure supports $2^{O(\sqrt{k})}$ update time and $O(k)$ query time for \textsc{$k$-Path/Cycle}.
For the other problems, our data structures support $O(\log n)$ update time and $2^{O(\sqrt{k})}$ query time, where $k$ denotes the output size.
\end{abstract}
\section{Introduction}

For a set $V$ of $n$ points in the plane, the \emph{unit disk graph} of 
$V$ is the intersection graph of the unit disks of diameter one centered at the points in 
$V$, denoted by $\ud(V)$. 
Unit disk graphs serve as a powerful model for real-world applications such as broadcast networks~\cite{kuhn2004unit,kuhn2003ad},  biological networks~\cite{HAVEL1983383} and facility location~\cite{wang1988study}.  
Due to various applications, unit disk graphs and geometric intersection graphs have gained significant attention in computational geometry. 
Since most of the fundamental NP-hard problems remain NP-hard even in unit disk graphs, the study of NP-hard problems on unit disk graphs focuses on approximation algorithms and parameterized algorithms~\cite{an2021feedback,an2024eth, demaine2005subexponential, zehavi2021eth, wu2006minimum}. 
From the perspective of parameterized algorithms, the main focus
is to design subexponential-time parameterized algorithms for various problems on unit disk graphs. While such algorithms do not exist for general graphs unless ETH fails, lots of problems admit subexponential-time parameterized algorithms for unit disk graph and geometric intersection graphs~\cite{an2021feedback,an2024eth,an2024sparseouterstringgraphslogarithmic, demaine2005subexponential,zehavi2021eth,lokshtanov2022subexponential}. 

In this paper, we study fundamental graph problems on \emph{dynamic} unit disk graphs. Given a vertex set $V$ that changes under vertex insertions and deletions, our goal is to maintain data structures so that specific problems for $\ud(V)$ can be solved efficiently. 
This dynamic setting has attracted considerable interest. For instance, the connectivity problem~\cite{chan_et_al:LIPIcs.SoCG.2024.36,chan2011dynamic}, 
the coloring problem~\cite{krawczyk2017line}, the independent set problem~\cite{bhore2023fully,erlebach2005polynomial}, the set cover problem~\cite{agarwal2022dynamic}, and the vertex cover problem~\cite{bhore2024fully} have been studied for dynamic geometric intersection graphs. 
Here, all problems, except for the connectivity problem, are NP-hard. 
All previous work on dynamic intersection graphs
for those problems study approximation algorithms. 
However, like the static setting, parameterized algorithms are also a successful approach for addressing NP-hardness in the dynamic setting. There has been significant research on parameterized algorithms for dynamic general graphs~\cite{alman2020dynamic,chen2021efficient,dvovrak2014dynamic,korhonen2023dynamic}. Surprisingly, however, there have been no studies on parameterized algorithms for dynamic unit disk graphs.

In this paper, we initiate the study of fundamental \emph{parameterized} problems on dynamic unit disk graphs. In particular, we study the following five fundamental problems in the dynamic setting. 
All these problems are NP-hard even for unit disk graphs~\cite{clark1990unit,itai1982hamilton}. 
\begin{itemize}
\item \textsc{$k$-Path/Cycle} asks to find a path/cycle of $G$ with  exactly $k$ vertices,
\item \textsc{$k$-Vertex Cover} asks to find a set $S$ of $k$ vertices s.t. $G\setminus S$ has no edge,
\item \textsc{$k$-Triangle Hitting Set} asks to find a set $S$ of  $k$ vertices s.t. $G\setminus S$ has no triangle,
\item \textsc{$k$-Feedback Vertex Set} asks to find a set $S$ of $k$ vertices s.t. $G\setminus S$ has no cycle, and
\item \textsc{$k$-Cycle Packing} asks to find $k$ vertex-disjoint cycles of $G$.
\end{itemize}
In the course of vertex updates, we are asked to solve those problems as a query. 
Except for \textsc{$k$-Path/cycle}, a query is given with an integer $k$. On the other hand, our data structure for \textsc{$k$-Path/Cycle} uses $k$ in the construction time.

\begin{table*}
    \resizebox{\textwidth}{!}
    {
    \begin{tabular}{c||c|c|c}
        \noalign{\smallskip}\hline
        \textbf{} & {Update time} & {Query time} &{Space Complexity} \\
        \hline
        \textsc{$k$-Path/Cycle*} & {
 ${2^{O(\sqrt k)}}$
 } & {$O(k)$} &\textbf{$O(kn)$} \\
        \hline
        \textsc{$k$-Vertex Cover*} & $O(1)$ & {$2^{O(\sqrt k)}$} &\textbf{$O(n)$}\\
        \hline
        \textsc{$k$-Triangle Hitting Set*} & $O(1)$ & {$2^{O(\sqrt k)}$} &\textbf{$O(n)$} \\
        \hline
        \textsc{$k$-Feedback Vertex Set} & {$O(\log n)$} & {$2^{O(\sqrt k)}$} &\textbf{$O(n)$} \\
        \hline
        \textsc{$k$-Cycle Packing} & {$O(\log n)$} & {$2^{O(\sqrt k)}$} &\textbf{$O(n)$} \\
        \hline
    \end{tabular}}
    \caption{
    \small
    Summary of our results.
    The results marked as * support amortized update times, and the others are worst-case update/query times.
    Except for \textsc{$k$-Path/Cycle}, no data structure requires $k$ in the construction time; The parameter $k$ is given as a query. 
    Additionally, the data structure for \textsc{$k$-Path/Cycle} can answer a decision query in constant time.
    }\label{table:summary}
     \vspace*{-1\baselineskip}
\end{table*}

\subparagraph*{Our results.} Our results are summarized in Table~\ref{table:summary}. Note that these are almost ETH-tight.
To see this, recall that no problem studied in this paper admits a $2^{o(\sqrt k)}n^{O(1)}$-time algorithm in the static setting unless ETH fails~\cite{de2020framework, fomin2019finding, fomin2012bidimensionality}.
Thus for any data structure for these problems on dynamic unit disk graphs with update time $T_u(n,k)$ and query time $T_q(n,k)$,
we must have $n\cdot T_u(n,k) +T_q(n,k) = 2^{\Omega(\sqrt k)}n^{O(1)}$ unless ETH fails. In particular, $n\cdot T_u(n,k) +T_q(n,k)$ is the time
for solving the static problem using dynamic data structures; we insert the vertices one by one and then answer the query. 
In our case, this static running time is $2^{O(\sqrt k)} n$ for \textsc{$k$-Path/Cycle}, $2^{O(\sqrt k)} + O(n)$ 
for \textsc{$k$-Vertex Cover} and \textsc{$k$-Triangle Hitting Set},
and $2^{O(\sqrt k)} + O(n\log n)$ for \textsc{$k$-Feedback Vertex Set} and \textsc{$k$-Cycle Packing}.
Interestingly, as by-products, we slightly improve the running times
of the best-known static algorithms in~\cite{an2021feedback,an2024eth} for 
\textsc{$k$-Feedback Vertex Set} and \textsc{$k$-Cycle Packing} on unit disk graphs from 
$2^{O(\sqrt k)}n^{O(1)}$ to $2^{O{(\sqrt k)}}+O(n\log n)$.\footnote{Moreover, they can be improved further to take $2^{O{(\sqrt k)}}+O(n)$ time with a slight modification.}

A main tool used in this paper is \emph{kernelization}, a technique compressing an instance of a problem into a small-sized equivalent instance called a \emph{kernel}. 
Kernelization is one of the fundamental techniques used in the field of parameterized algorithms~\cite{cygan2015parameterized}.
We use the same framework for all problems, except for \textsc{$k$-Path/Cycle}: 
For each update, we maintain kernels for the current unit disk graph. Given a query, it is sufficient to solve the problem on the kernel instead of the entire unit disk graph. 
Precisely, if the kernel size exceeds a certain bound, we immediately return a correct answer. Otherwise, the kernel size is small, say $O(k)$, and thus we can answer the query by applying the static algorithms on the kernel.

\subparagraph*{Related work.}
While dynamic parameterized problems on unit disks have not been studied before, 
static parameterized algorithms have been widely studied for unit disk graphs. 
For instance, Fomin et al.~\cite{fomin2019finding} presented $2^{O(\sqrt{k}\log k)}n^{O(1)}$-time algorithms for \textsc{$k$-Path/Cycle}, \textsc{$k$-Vertex Cover}, \textsc{$k$-Feedback Vertex Set}, and \textsc{$k$-Cycle Packing} problems on static unit disk graphs.
Subsequently, the running times were improved to $2^{O(\sqrt k)}(n+m)$~\cite{an2021feedback,an2024eth,cygan2015parameterized,de2020framework,zehavi2021eth}, which are all ETH-tight.
Additionally, for disk graphs, subexponential time FPT algorithms were studied for \textsc{$k$-Vertex Cover} and \textsc{$k$-Feedback Vertex Set}~\cite{an2023faster,lokshtanov2022subexponential}.

On the other hand, there are several previous works on dynamic parameterized problems on \emph{general graphs}.
Alman et al.~\cite{alman2020dynamic} presented a dynamic algorithm for \textsc{$k$-Vertex Cover} supporting $1.2738^{O(k)}$ query time and $O(1)$ amortized update time.
They also presented a dynamic algorithm for \textsc{$k$-Feedback Vertex Set} supporting $O(k)$ query time and $k^{O(k)}\log^{O(1)} n$ amortized update time.
Korhonen et al.~\cite{korhonen2023dynamic} presented a dynamic algorithm for \textsf{CMSO} testing, parameterized by treewidth.
Chen et al.~\cite{chen2021efficient} and Dvo{\v{r}}{\'a}k et al.~\cite{dvovrak2014dynamic} presented a dynamic algorithm for \textsc{$k$-Path/Cycle} and for \textsf{MSO} testing, respectively, parameterized by treedepth.
These algorithms admit \emph{edge} insertions and deletions, and Dvo{\v{r}}{\'a}k et al.~\cite{dvovrak2014dynamic} also admits \emph{isolated vertex} insertions and deletions, while we deal with vertex insertions and deletions.

An alternative way for dealing with NP-hardness is using approximation.
There are numerous works on approximation algorithms for dynamic intersection graphs. 
For disks, one can maintain $(1+\varepsilon)$-approximation of \textsc{Vertex Cover}~\cite{bhore2024fully,hochbaum1985approximation}.
Bhore et al.~\cite{bhore2023fully} presented a constant-factor approximation algorithm for the maximum independent set problem for disks.
They generalized their result on comparable-sized fat object graphs with the approximation factor depending on a given dimension and fatness parameter.
For intervals and unit-squares, Agarwal et al.~\cite{agarwal2022dynamic} presented a constant-approximation algorithm for the set cover problem and the hitting set problem.

\section{Preliminaries}
Throughout this paper, we let $V$ be a set of points in the plane, and we let $\ud(V)$ be the \emph{unit disk graph} of $V$. 
We interchangeably denote $v\in V$ as a point or as a vertex of $\ud(V)$ if it is clear from the context.
For an undirected graph $G$, 
we often use $V(G)$ and $E(G)$ to denote the vertex set of $G$ and the edge set of $G$, respectively. 
For convenience, we denote the subgraph of $G$ induced by $V(G)\setminus U$ by $G\setminus U$ for a subset $U$ of $V(G)$.  

\subparagraph*{Tree decomposition and weighted treewidth}
A \emph{tree decomposition} of a graph $G=(V,E)$ is a pair $(T,\mathcal B)$, where $T$ is a tree and $\mathcal B$ is a collection of subsets of $V$, called \emph{bags} of which each subset is associated with the nodes of $T$ with the following properties.
\begin{itemize}
    \item \textbf{(T1)} For any vertex $u\in V$, there is at least one bag in $\mathcal B$ which contains $u$,
    \item \textbf{(T2)} For any edge $uv\in E$, there is at least one bag in $\mathcal B$ which contains both $u$ and $v$,
    \item \textbf{(T3)} For any vertex $u\in V$, the subset of bags of $\mathcal B$ containing $u$ forms a subtree of $T$.
\end{itemize}
The \emph{width} of a tree decomposition is the size of its largest bag minus one, and the \emph{treewidth} of $G$ is the minimum width of a tree decomposition of $G$.
Specifically, when $G$ is a vertex-weighted graph, the \emph{weighted width} of the tree decomposition is defined as the maximum weight of the bags, where the weight of a bag is defined as the sum of the weights of the vertices in the bag.
Subsequently, the \emph{weighted treewidth} of graph $G$ is the minimum weighted width among the tree decomposition of $G$.

In this paper, we use the concept of the \emph{clique-based weighted width} introduced by de Berg et al.~\cite{de2020framework}.
For an unweighted graph $G$, let $\mathcal P$ be a decomposition of the vertex set $V(G)$ where each $P\in \mathcal P$ forms a clique in $G$.
Additionally, let $G_{\mathcal P}$ be the vertex-weighted graph obtained from $G$ by contracting each $P\in \mathcal P$ into a vertex $v_P$ and defining the weight of $v_P$ as $1+\log (|P|+1)$.
A tree decomposition $(T,\mathcal B)$ of $G$ is said to be a \emph{clique-based tree decomposition} with respect to $\mathcal P$ if it can be obtained from a tree decomposition $(T',\mathcal B')$ of $G_{\mathcal P}$ by uncontracting every vertex $v_P$ into the clique $P\in \mathcal P$.
Additionally, we refer to the weighted width of $(T',\mathcal B')$ as the \emph{clique-based weighted width} of $(T,\mathcal B)$.

\begin{figure}
    \centering
    \includegraphics[width=0.85\textwidth]{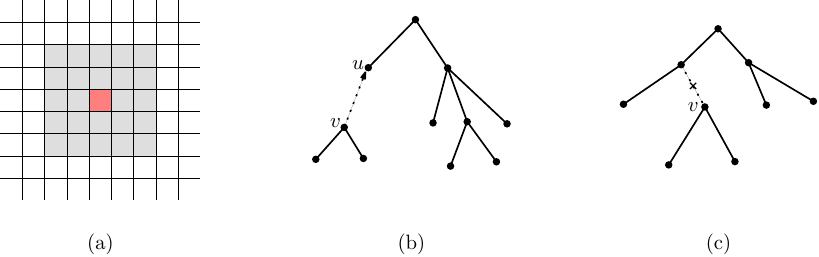}
    \caption{(a) The gray and red grid cells are $2$-neighboring cells of the red grid cell. (b) Illustration of $\textsf{Link}(u,v)$. (c) Illustration of $\textsf{Cut}(v)$. }
    \label{fig:neighbor-link-cut}
\end{figure}

\subparagraph*{Grid.}
A grid $\boxplus$ is a partition of the plane into squares (called grid cells) of diameter one.
Notice that any two points of $V$ contained in the same grid cell of $\boxplus$ are adjacent in $G$.
Each grid cell $\Box$ has its own id: $(\lfloor a/\sqrt{2} \rfloor, \lfloor b/\sqrt{2} \rfloor)$, where $a$ and $b$ are the $x$- and $y$-coordinates of a point in $\Box$.
For any two grid cells $\Box, \Box'$ with id $(x,x')$ and $(y,y')$, respectively, we let $d(\Box, \Box')=\max(|x-x'|, |y-y'|)$.
For an integer $\ell> 0$, a grid cell $\Box$ is called an \emph{$\ell$-neighboring cell} of a grid cell $\Box'$ if $d(\Box,\Box')\leq \ell$.
See Figure~\ref{fig:neighbor-link-cut}(a). 
We slightly abuse the notation so that $\Box$ itself is an $\ell$-neighboring cell of $\Box$ for all $\ell>0$.
For a point $v$ in the plane, we use $\Box_v$ to denote the cell of $\boxplus$ containing $v$. If it lies on the boundary of a cell,
$\Box_v$ denotes an arbitrary cell containing $v$ on its boundary. 

We do not construct the grid $\boxplus$ explicitly. Instead, 
we maintain the grid cells containing vertices of $V$ only. 
We associate each id with a linked list that stores all the vertices of $V$ contained in the grid cell.
Once we have the id of $\Box$, we fetch the linked list associated with $\Box$ in amortized constant time or  $O(\log n)$ worst-case time, where $n$ is the number of points of $V$. 
More specifically, this can be implemented in $O(1)$ amortized time using dynamic perfect hashing (once true randomness is available) and in $O(\log n)$ worst-case time using 1D range search tree along the lexicographic ordering of the ids.
Also, each update of the linked list can be done in the same time bound.

In this paper, we access grid cells only for updating data structures, and then we store the necessary grid cells explicitly in our data structures. 
Consequently, the update times of our data structures are sometimes analyzed using amortized analysis while the query times are always analyzed using worst-case analysis.

\subparagraph*{Link-cut tree.} When we update data structures, we use link-cut trees. 
A \emph{link-cut tree} is a dynamic data structure that maintains a collection of vertex-disjoint rooted trees and supports two kinds of operations: a link operation that combines two trees into one by adding an edge, and a cut operation that divides one tree into two by deleting an edge~\cite{sleator1981data}. 
See Figure~\ref{fig:neighbor-link-cut}(b--c).
Each operation requires $O(\log n)$ time. More precisely, the data structure supports the following query and update operations in $O(\log n)$ time. 
\begin{itemize} 
    \item $\textsf{Link}(u,v)$:
    If $v$ is the root of a tree and $u$ is a vertex in another tree, link the trees containing $v$ and $u$ by adding the edge between them, making $u$ the parent of $v$.
    \item $\textsf{Cut}(v)$: If $v$ is not a root, this removes the edge between $v$ and its parent, so that the tree containing $v$ is divided two trees containing either $v$ or not. 
    \item $\textsf{Evert}(v)$: This turns the tree containing $v$ “inside out” by making $v$ the root of the tree.
    \item $\textsf{Connected}(u,v)$: This checks if $u$ and $v$ are contained in the same tree. 
    \item $\textsf{LCA}(u,v)$: This returns the lowest common ancestor of $u$ and $v$ assuming that $u$ and $v$ are contained in the same tree.
    \item $\textsf{Root}(u)$: This returns the root of
    the tree containing $u$.
\end{itemize}

\section{Dynamic $k$-Path/Cycle Problem}\label{apx:long_path}
In this section, we describe a fully dynamic data structure for maintaining a path/cycle of length $k$ (if one exists) of $\ud(V)$ under \emph{vertex} insertions and deletions. 
Each query asks to return a $k$-path/cycle of $\ud(V)$ if it exists. 
This data structure supports $2^{O(\sqrt k)}$ update time and $O(k)$ query time.
Moreover, if each query asks to simply check if $V$ has 
a $k$-path/cycle, then we can answer this in $O(1)$ time.

We use the following two observations. First, if a grid cell contains at least $k$ vertices of $V$,
$\ud(V)$ has a $k$-path/cycle. In this case, we return arbitrary $k$ vertices contained in the same grid cell. Second, if $\ud(V)$ has a $k$-path/cycle $\pi$ containing a vertex $x\in V$, 
then all vertices of $\pi$ are contained in $V\cap \boxplus_x$, where $\boxplus_x$ denotes the union of $2k$-neighboring cells of the grid cell $\Box_x$ containing $x$.
Notice that 
$\ud(V\cap \boxplus_x)$ contains a $k$-path/cycle if $ |V\cap \boxplus_x|\geq k(4k+1)^2$
due to the first observation. 
We state these ideas formally as follows. 
\begin{lemma} \label{lem:lpc-kernel}
$(\ud(V), k)$ is a \textbf{yes}-instance of \textsc{$k$-Path/Cycle} if there is a vertex $x\in V$ with $|V\cap \boxplus_x|\geq f(k)$, where $f(k)=k(4k+1)^2$. 
\end{lemma}

\subparagraph*{Data structures.}
Recall that once we know the id of $\Box$, we can fetch the linked list storing $V\cap \Box$ in {$O(1)$ amortized time.} 
For a cell $\Box$, we define the $2k$-\emph{grid cluster} of $\Box$ as the union of $2k$-neighboring cells of $\Box$. 
In the course of updates of $V$, we maintain the following information 
for each grid cell $\Box$ with $V\cap \Box \neq\emptyset$ as data structures. 
\begin{itemize} 
    \item The number of vertices of $V$ contained in $\Box$,
    \item The number of vertices of $V$ contained in the $2k$-grid cluster of $\Box$,
    \item Value $\textsf{Bool}(\Box)\in\{\textbf{yes},\textbf{no}\}$ indicating that $\ud(V)$
    has a $k$-path/cycle contained in the $2k$-grid cluster of $\Box$, and
    \item A $k$-path/cycle $\pi(\Box)$ of $\ud(V)$ contained in the $2k$-grid cluster of $\Box$ if it exists.
\end{itemize}

We store this information by associating it with the cell id: now each cell $\Box$ is associated with not only the linked list for $V\cap \Box$ but also the four items stated above.
Furthermore, we globally maintain a linked list $\mathcal I$ storing the ids of the grid cells $\Box$ with $\textsf{Bool}(\Box)=\textbf{yes}$. 
The space complexity of our data structure is $O(kn)$ as it explicitly stores $\pi(\Box)$ for each $\Box$.

\subparagraph*{Update algorithm.}
Suppose we have the above information, and we aim to update $V$. Let $v$ be the vertex
to be removed from $V$ (or added to $V$). Notice that we are required to update the four items we maintain for the $2k$-neighboring cells of $\Box_v$ including $\Box_v$ itself.
We get the id of $\Box_v$ and update the number of vertices in $\Box_v$. Next, we enumerate the ids of all $2k$-neighbors $\Box$ of $\Box_v$ and update the number of vertices in the $2k$-grid cluster for $\Box$ along with $\Box_v$.
Since the number of $2k$-neighboring cells of $\Box_v$ is $O(k^2)$, this takes $O(k^2)$ time.
Before we update the third and fourth items, we handle the exception where 
$|V\cap \Box_v|$ becomes 0 from 1. In this case, we do not need to maintain the four items for $\Box_v$ any longer, so we dissociate them from $\Box_v$. 
Additionally, if $\textsf{Bool}(\Box_v)$ was previously \textbf{yes}, we remove the id of $\Box_v$ from $\mathcal I$. 
Then we use the algorithm of~\cite{zehavi2021eth} for updating the third and fourth items maintained for $\Box_v$ and its $2k$-neighboring cells.
\begin{lemma}[Theorem 1 in~\cite{zehavi2021eth}] \label{lem:lpc-blackbox}
    We can compute a $k$-path/cycle in a unit disk graph with $n$ vertices in $2^{O(\sqrt k)}n^{O(1)}$ time.
\end{lemma}

For each $2k$-neighboring cell $\Box$ of $\Box_v$ with $V\cap\Box\neq\emptyset$ along with $\Box_v$ itself, we update $\textsf{Bool}(\Box)$ as follows.
If the number of vertices in the $2k$-neighboring cells of $\Box$ is at least $f(k)$, we set $\textsf{Bool}(\Box)$ to \textbf{yes} due to Lemma~\ref{lem:lpc-kernel}, where $f(k)=k(4k+1)^2$.
Otherwise, we compute $\textsf{Bool}(\Box)$ in $2^{O(\sqrt k)}f(k)^{O(1)}=2^{O(\sqrt k)}$ time by running the algorithm of Lemma~\ref{lem:lpc-blackbox} with the unit disk graph induced by the vertices contained in the $2k$-grid cluster of $\Box$.
Therefore, we can update $\textsf{Bool}(\Box)$ for all $2k$-neighboring cells $\Box$ of $\Box_v$ in 
$2^{O(\sqrt k)}\cdot O(k^{2})=2^{O(\sqrt k)}$ time in total. 
Moreover, we can also compute $\pi(\Box)$ simultaneously.
Finally, for each cell $\Box$ whose value $\textsf{Bool}(\Box)$ has been changed from \textbf{yes} (and \textbf{no}) to \textbf{no} (and \textbf{yes}), we remove the id of $\Box$ from $\mathcal I$ (and insert the id of $\Box$ to the tail of $\mathcal I$) in $O(1)$ time.

\begin{lemma}
    The update algorithm takes $2^{O(\sqrt k)}$ time in total.
\end{lemma}

\subparagraph*{Query algorithm.}
We simply answer \textbf{yes} if and only if $\mathcal I$ is not empty. 
This takes $O(1)$ time. 
Furthermore, if the answer is \textbf{yes}, we obtain an explicit solution by referencing the id of $\Box$ stored in the head of $\mathcal I$ and retrieve $\pi(\Box)$ accordingly in
$O(k)$ time, which is linear in the output size. 
The following lemma verifies the correctness. 

\begin{lemma}
    $(\ud(V), k)$ is a \textbf{yes}-instance of \textsc{$k$-Path/Cycle} iff $\mathcal I$ is not empty.
\end{lemma}
\begin{proof}
    $(\ud(V),k)$ is a \textbf{yes}-instance if and only if there is a vertex $x\in V$ contained in a $k$-path/cycle.
    This is equivalent to the case that $\textsf{Bool}(\Box_x)$ is \textbf{yes}.
    Since $\mathcal I$ stores the ids of all grid cells $\Box$ with $\textsf{Bool}(\Box)=\textbf{yes}$, the two conditions of the lemma are equivalent.
\end{proof}
\begin{theorem} \label{thm:lpc-result}
    For a fixed integer $k$, there is an $O(kn)$-sized fully dynamic data structure on the unit disk graph induced by a vertex set $V$ 
    supporting $2^{O(\sqrt k)}$ update time that answers each $k$-path/cycle reporting query in $O(k)$ time, and each decision query in $O(1)$ time.
\end{theorem}

\section{Dynamic Vertex Cover Problem}\label{apx:vertex_cover}
In this section, we describe a fully dynamic data structure on the unit disk graph of a vertex set $V$ dynamically changing under vertex insertions and deletions that can answer
vertex cover queries efficiently. 
A query is given with a positive integer $k$ and aims to find a vertex cover of $\ud(V)$ of size at most $k$.
For this, we maintain the vertex set $V'\subseteq V$ which is the union of non-isolated vertices in $\ud(V)$ and the vertices $v$ where $\Box_v$ is a 5-neighboring cell of the grid cells containing at least two vertices of $V$.
See Figure~\ref{fig:VC_kernel}.
It is trivial that $(\ud(V),k)$ is a \textbf{yes}-instance if and only if $(\ud(V'),k)$ is a \textbf{yes}-instance.

\begin{figure}
    \centering
    \includegraphics[width=0.65\textwidth]{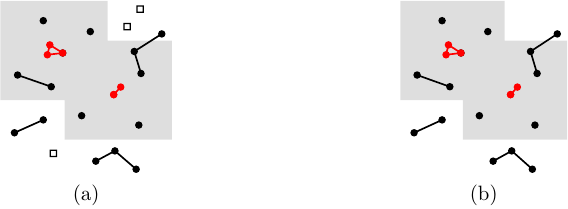}
    \caption{(a) Illustration of $\ud(V)$. The gray part represents a 5-neighboring cell of the grid cells containing at least two vertices. These vertices are colored in red. Square vertices are vertices of  $V\setminus V'$. (b) Illustration of $\ud(V')$. }
    \label{fig:VC_kernel}
\end{figure}

\begin{lemma}\label{lem:query_vc}
    If $(\ud(V),k)$ is a \textbf{yes}-instance of $\vc$, $|V'|=O(k)$.
\end{lemma}
\begin{proof}
    Recall that $(\ud(V),k)$ is a \textbf{yes}-instance if and only if $(\ud(V'),k)$ is a \textbf{yes}-instance.
    Let $X\subset V'$ be a vertex cover of $\ud(V')$ of size $k$.
    Then $\ud(V')-X$ is a unit disk graph where the vertices are isolated.

    Let $\boxplus_{\geq 2}$ be the union of the 5-neighboring cells of the grid cells containing at least two vertices.
    Note that the vertices in the same grid cell form a clique.
    Thus a grid cell containing at least two vertices contains at least one vertex in $X$, and each grid cell can contain at most one vertex not in $X$.
    Thus the number of grid cells in $\boxplus_{\geq 2}$ is $O(k)$, and thus $\boxplus_{\geq 2}$ contains at most $O(k)$ vertices.
    Note that an endpoint of an edge is in a 5-neighboring cell of the grid cell containing the other endpoint.
    Thus the vertices not in $\boxplus_{\geq 2}$ are not covered by a vertex of $X$ in a grid cell containing at least two vertices.
    Note that such vertices of $X$ have a constant degree.
    Thus the number of vertices not in $\boxplus_{\geq 2}$ is $O(k)$ since $X$ is a vertex cover of size $k$.
    This implies that $|V'|=O(k)+|X|=O(k)$.
\end{proof}

To update $V'$ efficiently, we store the size of $\Box\cap V$ for each cell $\Box$  with $\Box\cap V\neq \emptyset$.

\subparagraph*{Update algorithm.}
Suppose that we already have a non-isolated vertex set $V'$ for the current vertex set $V$.
When we insert a vertex $v$, we first update the number of vertices in $\Box_v$.
After that, we need to update $V'$. 
For this, we access the ids of the 5-neighboring cells $\Box$ of $\Box_v$ containing $V \cup \{v\}$.
If one $\Box$ of the cells along with $\Box_v$ itself contains at least two vertices of $V$, then we insert the vertex $v$ into $V'$.
Additionally, if $\Box_v$ contains at least two vertices of $V$, we insert the vertices in the 5-neighboring cells of $\Box_v$ containing at most one vertex.
Note that the number of such vertices is $O(1)$.
Otherwise, only $O(1)$ vertices are contained in the 5-neighboring cells of $\Box_v$.
We insert $v$ into $V'$ when $v$ is adjacent to one of those $O(1)$ vertices in $\ud(V)$.
This takes $O(1)$ time.
Analogously, we can delete a vertex from $V$ in $O(1)$ time.


\subparagraph*{Query algorithm.}
Given an integer $k$ as a query, we aim to find a vertex cover of size $k$ of $\ud(V)$.
Recall that we maintain $V'$ for the current $V$, and it is sufficient to check if $(\ud(V'),k)$ is a \textbf{yes}-instance.
By Lemma~\ref{lem:query_vc}, if $|V'|$ is at least $\Omega(k)$, then we return \textbf{no} immediately.
Otherwise, we compute a vertex cover of $\ud(V')$ of size $k$ using the clique-based tree decomposition as follows.
We compute a tree decomposition of clique-based weighted treewidth $O(\sqrt{k})$ of $\ud(V')$ in $\text{poly}(k)$ time by applying the algorithms in~\cite{an2024eth,cygan2015parameterized,de2020framework}.
For each clique in $\ud (V')$, a vertex cover contains all but one vertex of the clique.
Thus, we can compute the minimum vertex cover of $\ud(V')$ (also in $\ud(V)$) in $2^{O(\sqrt k)}$ time using a standard dynamic programming algorithm on a tree decomposition of clique-based weighted treewidth $O(\sqrt k)$ as observed in~\cite{de2020framework}.

\begin{theorem}
    There is an $O(n)$-sized fully dynamic data structure on the unit disk graph induced by a vertex set $V$ supporting $O(1)$ update time that allows us to compute a
    vertex cover of size at most $k$ in $2^{O(\sqrt k)}$ time. 
\end{theorem}

\section{Dynamic Triangle Hitting Set Problem}\label{sec:hittings}
In this section, we describe a fully dynamic data structure on the unit disk graph of a vertex set $V$ dynamically changing under vertex insertions and deletions that can answer
triangle hitting set queries efficiently. 
Each query is given with a positive integer $k$ and 
asks to return a triangle hitting set of $\ud(V)$ of size at most $k$.
This data structure will also be used for \textsc{Feedback Vertex Set} and \textsc{Cycle Packing}
in Sections~\ref{sec:fvs} and~\ref{sec:cpp}.

Our strategy is to maintain a \emph{kernel} of $(\ud(V),k)$. 
More specifically, consider the set $V_\textsf{tri}$ of vertices contained in triangles of $\ud(V)$.
Then $(\ud(V_\textsf{tri}),k)$ is a \textbf{yes}-instance if and only if $(\ud(V),k)$ is a \textbf{yes}-instance, i.e., it is a kernel of 
$(\ud(V),k)$. We can show that the size of $V_\textsf{tri}$ is $O(k)$
if $(\ud(V),k)$ is a \textbf{yes}-instance. Therefore, it is sufficient to maintain
$V_\textsf{tri}$ for answering queries.
However, it seems unclear if $V_\textsf{tri}$ can be updated in $O(1)$ time, which is the desired update time.
In particular, imagine that a vertex $v$ is inserted to $V$, and we are to determine if $v$ is contained in a triangle of $\ud(V)$. In the case that two neighboring cells $\Box_1$ and $\Box_2$ of $\Box_v$ contain $\omega(k)$ vertices of $V$, we need to determine if there are vertices $x_1\in\Box_1$ and $x_2\in\Box_2$ such that $x_1, x_2$ and $v$ form a triangle of $\ud(V)$.
 
To overcome this issue, we use a superset of $V_\textsf{tri}$ as a kernel. 
Notice that as long as a subset of $V$ contains $V_\textsf{tri}$, it is a kernel of $(\ud(V),k)$.

\begin{figure}
    \centering
    \includegraphics[width=0.48\textwidth]{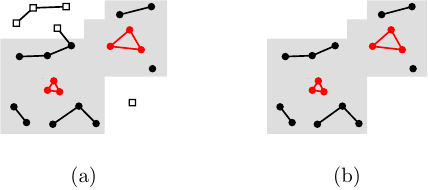}
    \caption{(a) Illustration of $\ud(V)$. The gray part represents the cells of $\boxplus_\irr$. Vertices of $V_\textsf{tri}$ are colored in red and square vertices are vertices of $V \setminus V_\irr$. (b) Illustration of $\ud(V_\irr)$.}
    \label{fig:V_tri_and_V_core}
\end{figure}

\subparagraph*{Kernel: $\boxplus_\irr$ and $V_\irr$.}
The \emph{core grid cluster}, denoted by $\boxplus_\irr$, is defined as the union of the 5-neighboring cells of the grid cells containing a vertex of $V_\textsf{tri}$ and the 10-neighboring cells of the grid cells
containing at least three vertices of $V$. Then let $V_\irr$ be the set of vertices of $V$ contained in $\boxplus_\irr$.
See Figure~\ref{fig:V_tri_and_V_core}.
Note that the degree of every vertex of $V\setminus V_\irr$ in $\ud(V)$ is $O(1)$.
We will use this property for designing the update algorithm.

\begin{lemma}\label{lem:irrducible_bound_fvs}
        The size of $V_{\irr}$ is $O(k)$ if $\ud(V)$ has at most $k$ vertex-disjoint triangles.
\end{lemma}
\begin{proof}
    We fix a maximum number of vertex-disjoint triangles of $\ud(V)$. Let $k'$ be the number of such triangles. 
    We first demonstrate that  
    there are $O(k')$ grid cells contained in $\boxplus_\irr$. 
    For this, observe that the vertices of the vertex-disjoint triangles are contained in $O(k')$ grid cells. 
    Moreover, a vertex of any triangle of $\ud(V)$ lies in a 5-neighboring  cell of such grid cells. Therefore, the number of grid cells containing vertices of $V_\textsf{tri}$ is $O(k')$. By definition, a grid cell contained in $\boxplus_\irr$  is a 10-neighboring cell of a grid cell containing a vertex of $V_\textsf{tri}$. 
    Therefore, the number of grid cells contained in $\boxplus_\irr$ is also $O(k')$.

    Now we focus on the number of vertices contained in $\boxplus_\irr$. 
    Note that for each grid cell $\Box$, all but at most two vertices of $V\cap \Box$ must be contained in the vertex-disjoint triangles.  
    Therefore, 
    the number of vertices of $V$ contained in $\boxplus_\irr$ but not contained in the vertex-disjoint triangles is $O(k')$, and the number of vertices of $V$ contained in the vertex-disjoint triangles is $3k'$. 
    In conclusion, $V_\irr$ has $O(k')$ vertices. Moreover, if $k'\leq k$, the size of $V_\irr$ is $O(k)$.
\end{proof}

\begin{observation}\label{obs:neighbors}
    Given a vertex $v\in V\setminus V_\irr$, we can compute its neighbors in $\ud(V)$ in $O(1)$ time. 
\end{observation}


\subparagraph*{Query algorithm.} 
Given an integer $k$ as a query, our goal is to check if $(\ud(V),k)$ is a \textbf{yes}-instance. Since $(\ud(V_\irr),k)$ is a kernel of $(\ud(V),k)$, it is sufficient to check if $(\ud(V_\irr),k)$ is a \textbf{yes}-instance. 
We first check if the size of $V_\irr$ is $O(k)$.
If it is not the case, $\ud(V_\irr)$ contains at least $k+1$ vertex-disjoint triangle by Lemma~\ref{lem:irrducible_bound_fvs}, and thus any triangle hitting set of $\ud(V_\irr)$
has a size larger than $k$. We return \textbf{no} immediately. 
If the size of $V_\irr$ is $O(k)$, the clique-based weighted treewidth of $\ud(V_\irr)$ is $O(\sqrt k)$ by Lemma~2 of~\cite{an2023faster}.
In this case, the smallest triangle hitting set of $\ud(V_\irr)$ can be computed in $2^{O(\sqrt k)}$ time as follows. 
We compute a tree decomposition of clique-based weighted treewidth $O(\sqrt{k})$ of $\ud(V_\irr)$ in $\text{poly}(k)$ time.
For a clique in $\ud(V_\irr)$, every triangle hitting set contains all but at most two vertices in the clique.
Thus we can compute the minimum triangle hitting set of $\ud(V)$ in $2^{O(\sqrt{k})}$  time 
using a standard dynamic programming algorithm on a tree decomposition of clique-based weighted treewidth $O(\sqrt k)$ as observed in~\cite{de2020framework}.

\subparagraph*{Update algorithm.}
Suppose that we already have $V_\irr$ for the current vertex set $V$.
Imagine that we aim to insert a vertex $v$ to $V$. Then we need to update $V_\irr$. 
For this, it suffices to determine if $\Box$ must be contained in $\boxplus_\irr$ for every 10-neighboring cell $\Box$ of $\Box_v$. 
By Observation~\ref{obs:irreducible_in_constant}, this takes $O(1)$ time. 
The deletion of a vertex $v$ from $V$ can be handled analogously in $O(1)$ time.

\begin{observation}\label{obs:irreducible_in_constant}
    We can check if a grid cell $\Box$ is contained in $\boxplus_\irr$ in $O(1)$ time.
\end{observation}
\begin{proof}
    Recall that $\Box$ is contained in $\boxplus_\irr$ 
    if and only if either one of its 10-neighboring cells contains at least three vertices of $V$ or one of 
    its 5-neighboring cells contains a vertex of $V_\textsf{tri}$. 
    We can check if $\Box$ belongs to the first case in $O(1)$ time  
    by enumerating the ids of all $10$-neighboring cells of $\Box$
    and fetching the linked lists stored for those cells. 
    If $\Box$ does not belong to the first case, all its 10-neighboring cells
    contain at most two vertices of $V$. 
    Therefore, we can enumerate all vertices contained in the 5-neighboring cells of $\Box$ in $O(1)$ time, and check if each of them is contained in a triangle of $\ud(V)$ in $O(1)$ time by enumerate all vertices contained in the 10-neighboring cells of $\Box$. In this way, we can determine if $\Box$ is contained in $\boxplus_\irr$ in $O(1)$ time. 
\end{proof}

\begin{theorem}\label{thm:ths_perf}
    There is an $O(n)$-sized fully dynamic data structure on the unit disk graph induced by a vertex set $V$ supporting $O(1)$ update time that allows us to compute a
    triangle hitting set of size at most $k$ in $2^{O(\sqrt k)}$ time. 
\end{theorem}

\section{Dynamic Feedback Vertex Set Problem} \label{sec:fvs}
In this section, we describe a fully dynamic data structure on the unit disk graph of a vertex set $V$ dynamically changing under vertex insertions and deletions that can answer
feedback vertex set queries efficiently. 
Each query is given with a positive integer $k$ and 
asks to return a feedback vertex set of $\ud(V)$ of size at most $k$.
As a data structure, we use the core grid cluster $\boxplus_\irr$ introduced in Section~\ref{sec:hittings}. In addition to this, we design a new data structure,
which will also be used for \textsc{Cycle Packing} in Section~\ref{sec:cpp}.

\begin{figure}
    \centering
    \includegraphics[width=0.72\textwidth]{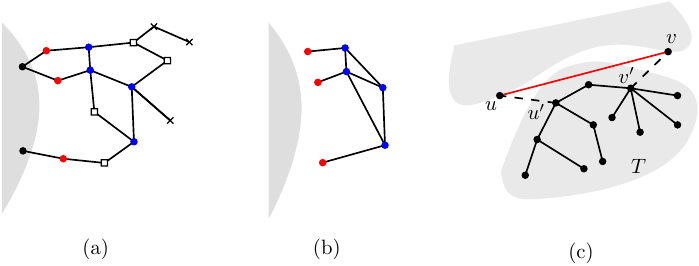}
    \caption{(a) Illustration of $\ud(V)$. The vertices of $\ud(V)$ in $V_\irr$, the boundary vertices, and the non-boundary vertices in $M$ are marked by the black, red, and blue vertices, respectively. The removed vertices and the contracted vertices in the construction of $M$ are marked by cross vertices and boxes, respectively. (b) Illustration of $M$. (c) In the construction of $M$, we are left with the induced path between $u$ and $v$, which will be contracted to the red edge in $M$. The two end edges of the path compose the bridge set of $T$.}
    \label{fig:M_construction_bridge}
\end{figure}

\subsection{Data Structure}
Note that a cycle of $\ud(V)$ might contain a vertex lying outside of $\boxplus_\irr$.
Thus we need to consider the part of $\ud(V)$ lying outside of $\boxplus_\irr$.
Since the complexity of $\ud(V\setminus V_\irr)$ can be $\Theta(n)$ in the worst case, 
we cannot afford to look at all such vertices to handle a query. Instead, we maintain a minor $M$ of $\ud(V\setminus V_\irr)$ of complexity $O(k)$, which will be called the \emph{skeleton} of $\ud(V\setminus V_\irr)$, such that
the graph obtained by gluing $\ud(V_\irr)$ and $M$ has a feedback vertex set of size $k$
if and only if $\ud(V)$ has a feedback vertex set of size $k$. 
Then it suffices to look at $V_\irr$ and $M$ to handle a query. 
Given an update of $V$, we need to update both $V_\irr$ and $M$. 
Since $M$ is a minor of $\ud(V\setminus V_\irr)$, some vertices of $\ud(V\setminus V_\irr)$ do not appear in $M$.
To handle the update of $V_\irr$ and $M$ efficiently, 
we construct an auxiliary data structure $\mathcal T$, which maintains 
the vertices of $\ud(V\setminus V_\irr)$ not appearing in $M$ efficiently.  

\subparagraph*{Skeleton $M$ of $\ud(V\setminus V_\irr)$.}
Given a query, we will glue $M$ and $\ud(V_\irr)$ together. For this purpose, we need to keep all vertices of $V\setminus V_\irr$ adjacent to a vertex of $V_\irr$ in the construction of $M$. We call 
such a vertex a \emph{boundary vertex}. 
We let $M$ be the graph obtained from $\ud(V\setminus V_\irr)$ by removing non-boundary vertices with the degree at most one in $M$ repeatedly, and then by 
contracting every maximal induced path consisting of only non-boundary vertices into a single edge. Note that the resulting graph $M$ is a minor of  $\ud(V\setminus V_\irr)$. 
Each vertex of $M$ corresponds to a vertex of $\ud(V\setminus V_\irr)$ of degree at least three or a boundary vertex. 
Furthermore, each edge of $M$ corresponds to an edge of $\ud(V)$ or an induced path of $\ud(V\setminus V_\irr)$.
See Figure~\ref{fig:M_construction_bridge}(a--b).
Note that $M$ is planar since every triangle-free disk graph is planar~\cite{breu1996algorithmic}.


\begin{lemma} \label{lem:fvs-m-bound}
If $(\ud(V),k)$ is a \textbf{yes}-instance of \textsc{Feedback Vertex Set}, $|V(M)|=O(k)$.
\end{lemma}
\begin{proof}
Recall that $M$ is a minor of $\ud(V)$.
Thus, $(M,k)$ is also a \textbf{yes}-instance for \fvs, if $(\ud(V),k)$ is a \textbf{yes}-instance of \textsc{Feedback Vertex Set}.
Let $X$ be a feedback vertex set of $M$ of size $k$.

We have that $M-X$ is a forest.
Since every vertex of $M$ of degree at most two is a boundary vertex of $M$, every leaf node of the forest $M-X$ is adjacent to a vertex of $X$ or $V_\irr$.
The number of such leaf nodes is at most $O(k)$ since every vertex of $M$ has constant degree and the number of boundary vertices in $M$ is $O(k)$. 
Furthermore, the number of vertices of degree at least three is linear in the number of leaf nodes of the trees, and thus, there are at most $O(k)$ vertices in $M$.
\end{proof}


\subparagraph*{Contracted forest $\mathcal T$ concerning $M$.}
We will see that it suffices to maintain $V_\irr$ and $M$ for the query algorithm.
However, to update $M$ efficiently, we need a data structure for the subgraph of $\ud(V)$ induced by $V\setminus (V_\irr\cup V(M))$. 
Note that the subgraph is a forest. We maintain the trees of this forest using the link-cut tree data structure. Let $\mathcal T$ denote the link-cut tree data structure. If it is clear from the context, we sometimes use $\mathcal T$ to denote the forest itself. 
Along with the link-cut trees, we associate each tree $T$ of $\mathcal T$ with a \emph{bridge set}. An edge of $\ud(V)$ incident to both $V(M)$ and $V\setminus (V_\irr\cup V(M))$ is called a \emph{bridge}. Then the bridge set of $T$
is the set of bridges incident to $T$.
See Figure~\ref{fig:M_construction_bridge}(c).

\begin{figure}
    \centering
    \includegraphics[width=0.45\textwidth]{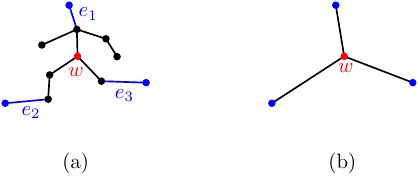}
    \caption{(a) A tree with 3 bridges $e_1$, $e_2$ and $e_3$. The vertex $w$ is the lowest common ancestor of the endpoints of $e_2$ and $e_3$. (b) The vertex $w$ has degree 3 after removing and contracting process, so it needs to appear in $M$. }
    \label{fig:branch_size_3}
\end{figure}

\begin{observation}\label{obs:edges-T-M}
    Each tree of $\mathcal T$ is incident to at most two bridges. 
\end{observation}
\begin{proof}
    Assume to the contrary that there is a tree $T$ of $\mathcal T$ incident to three bridges, say $e_1$, $e_2$ and $e_3$. 
    We consider the endpoint of $e_1$ in $T$ as the root of $T$.
    Let $w$ be the lowest common ancestor in $T$ between the endpoints of $e_2$ and $e_3$ in $T$.
    See Figure~\ref{fig:branch_size_3}. 
    There are at least three edge-disjoint paths from $w$ to some vertices of $M$. 
    This contradicts the definition of $M$ since the degree of $w$ is at least three even after removing all degree-1 vertices and contracting all maximal induced paths.
\end{proof}

\subsection{Query Algorithm}\label{sec:fvs-kernel}
In this subsection, we show how to compute a feedback vertex set of size $k$ of $\ud(V)$ in $2^{O(\sqrt{k})}$ time using $V_\irr$ and $M$ only.
Let $G$ be the graph obtained by gluing $\ud(V_\irr)$ and $M$.
More precisely, the vertex set of $G$ is the union of $V_\irr$ and $V(M)$.
There is an edge $uv$ in $G$ if and only if $uv$ is either an edge of $\ud(V_\irr)$, an edge of $M$, or an edge of $\ud(V)$ between $u\in V_\irr$ and $v\in V(M)$.
Notice that $v$ is a boundary vertex of $M$ in the third case.
By the following lemma, we can compute a feedback vertex set of $\ud(V)$ of size $k$ by computing a feedback vertex set of $G$.

\begin{lemma}
    $G$ has a feedback vertex set of size $k$ if and only if $\ud(V)$ has a feedback vertex set of size $k$.
\end{lemma}
\begin{proof}
    We first prove the forward direction. 
    Let $X$ be a feedback vertex set of $G$ of size $k$. 
    Recall that $\ud(V_\irr)$ is a subgraph of $\ud(V)$, and a vertex in $M$ corresponds to a vertex of $\ud(V)$.
    Thus, $X$ is a subset of $V$, and then it suffices to show that $X$ hits every cycle of $\ud(V)$.
    Let $C$ be a cycle of $\ud(V)$.
    If $C$ consists of vertices in $G$, it is hit by $X$ by the choice of $X$.
    Otherwise, each maximal subpath of $C$ consisting of vertices of $V\setminus V(M)$ is contracted into an edge in $M$, and thus there is a cycle $C'$ in $M$ consisting of the edges of $C$ and the contracted edges.
    By definition, $X$ hits $C'$. Since $V(C')\subseteq V(C)$, $X$ also hits $C$.
    Thus, $X$ is a feedback vertex set of $\ud(V)$ of size $k$.
    
    For the other direction, let $X$ be a minimum feedback vertex set of $\ud(V)$.  
    If a vertex $v$ in $X$ is not in $G$, 
    there are two cases: while removing all degree-1 vertices in the construction of $M$, $v$ becomes a degree-1 vertex at some point, 
    and thus it is removed from the graph. 
    Or, after removing all degree-1 vertices, it is
    contained in a maximal induced path $\pi$ in the remaining graph. 
    For the first case, $v$ cannot be contained in any cycle of $\ud(V)$, and thus
    $X\setminus \{v\}$ is also a feedback vertex set of $\ud(V)$, which violates the choice of $X$. For the second case, 
    the endpoints of $\pi$ are contained in $M$ by the construction of $M$.
    Let $u$ be an endpoint of $\pi$. 
    Since $\pi$ is a maximal induced path of the graph obtained from $\ud(V)$ by removing every degree-1 vertex repeatedly, a cycle of $\ud(V)$ hit by $v$ contains $\pi$, and thus $u$ hits every cycle hit by $v$.
    Then, $X'=(X\setminus\{v\})\cup \{u\}$ is a feedback vertex set of $\ud(V)$. 
    Furthermore, $X'$ is a feedback vertex set of $G$ with $|X|=|X'|$.
\end{proof}

We apply the algorithm proposed by An and Oh~\cite{an2021feedback} to compute a feedback vertex set of $G$ of size $k$ in $2^{O(\sqrt{k})}$ time.
The algorithm of~\cite{an2021feedback} computes a feedback vertex set of size $k$ of a unit disk graph with $n$ vertices in $2^{O(\sqrt k)}n^{O(1)}$ time using the clique-based tree decomposition.
In our case,
$G$ is not necessarily a unit disk graph. Thus we need to modify the algorithm of~\cite{an2021feedback} slightly as follows. 

First, as in~\cite{an2021feedback}, we compute a \emph{grid-based contraction} of $G$ as follows. 
We draw $G$ as in $\ud(V)$. More specifically, we put the vertices of $G$ on the points where their corresponding vertices in $\ud(V)$ lie. An edge $uv$ of $G$ corresponds to either an edge of $\ud(V)$ or a path of $\ud(V)$. We draw $uv$ as the curve in the drawing of $\ud(V)$ representing the path (or the edge) of $\ud(V)$ corresponding to $uv$. 
For each grid cell $\Box$ containing a vertex of $V_\irr$, we contract
all vertices in $\Box$ of $G$ into a single vertex. Here, we put the new vertex on an arbitrary point in $\Box$. The resulting graph is not necessarily planar
as two edges incident to contracted vertices might cross. To make the graph planar,
we add all crossing points in the drawing as vertices. Let $G'$ be the resulting graph, and we call it the \emph{grid-based contraction} of $G$. 
For each vertex $v$ of $G'$, we define its weight as the sum of $1+\log (|\Box\cap V(G)|+1)$ for all 10-neighboring cells $\Box$ of $\Box_v$. 


\begin{lemma}\label{lem:fvs-G'-weighted}
    The weighted treewidth of $G'$ is $O(\sqrt k)$ if $G$ has a feedback vertex set of size $k$.
    Moreover, in this case, we can compute a tree decomposition of $G'$ of weighted width $O(\sqrt k)$ in $2^{O(\sqrt k)}$ time. 
\end{lemma}
\begin{proof}
    We first prove that the weighted treewidth of $G'$ is $O(\sqrt{k})$ if $G$ has a feedback vertex set of size $k$.
    It suffices to show that there is a balanced separator of $G'$ of weight $O(\sqrt{k})$.
    Then, this implies that there is a tree decomposition of $G'$ of weighted width $O(\sqrt{k})$ since $G'$ is a planar and $G'$ can be recursively decomposed by balanced separators of weight small.
    
    Note that $G'$ is a planar graph, and thus a balanced separator of $G'$ has a weight of $O(\sqrt{\sum_{v_\in V(G')} w(v)^2})$ by~\cite{djidjev1997weighted}, where $w(v)$ be the weight of $v$ in $G'$.
    Recall that $w(v)$ is $\sum_\Box (1+ \log(|\Box\cap V(G)|+1))$, where in the summation, we take all $10$-neighboring cells $\Box$ of $\Box_v$.
    By the Cauchy-Schwarz inequality, the square of $\sum_\Box (1+ \log(|\Box\cap V(G)|+1))$ is at most $O(1)\cdot (1+\sum_\Box (\log(|\Box\cap V(G)|+1))^2)$.
    Observe that each grid cell contains at most a constant number of crossing points since a vertex of $G'$ in a grid cell $\Box$ can be adjacent to a vertex of $G'$ in $O(1)$-neighboring cells of $\Box$.
    Thus, for each cell $\Box$, it contributes the weights of $O(1)$ vertices of $G'$.
    Then, the squared sum of $w(v)$ is at most $|V(G')| + \sum O(\log(|\Box\cap V(G)|+1))^2$, where in the summation, we take all grid cells $\Box$ containing a vertex of $G'$.
    Since $(\log (z + 1))^2 \leq 2z$ for every positive integer $z$, it is at most $|V(G')|+\sum O(|\Box \cap V(G)|)$, and it is $O(k)$ by Lemma~\ref{lem:irrducible_bound_fvs} and Lemma~\ref{lem:fvs-m-bound}.
    Thus a balanced separator of $G'$ has a weight of $O(\sqrt{k})$.

    We can compute a tree decomposition of $G'$ of weighted width $O(\sqrt{k})$ in $2^{O(\sqrt{k})}$ time utilizing the algorithm proposed by de Berg et al.~\cite{de2020framework}.
    This algorithm computes a tree decomposition of weighted width $O(\omega)$ in $2^{O(\omega)}n^{O(1)}$ time, where $\omega$ is the weighted treewidth and $n$ is the size of a given graph.
    Note that for $G'$ $\omega=O(\sqrt{k})$ and $n=O(k)$.
    In the remainder of this proof, we illustrate the sketch of the algorithm.
    We compute an unweighted graph $G'_\textsf{un}$ by replacing each vertex of $G'$ with a clique of size $\lceil w(v) \rceil$, and then, we compute a tree decomposition $(T_\textsf{un}, \mathcal B_\textsf{un})$ of $G'_\textsf{un}$ by applying a constant-factor approximation algorithm for an optimal tree decomposition of~\cite{cygan2015parameterized}.
    Then, we compute a tree decomposition $(T_\textsf{un}, \mathcal B)$ of $G'$ of weighted width $O(\sqrt{k})$.
    Precisely, we set every bag of $\mathcal B$ to be empty, and add a vertex $v$ of $G'$ into a bag of $\mathcal B$ if its corresponding bag of $\mathcal B_\textsf{un}$ contains the clique corresponding to $v$.
    %
\end{proof}


Now, we illustrate how to construct the tree decomposition of $G$ of clique-based weighted width $O(\sqrt{k})$ from $(T',\mathcal B')$, where $(T',\mathcal B')$ is a tree 
decomposition of $G'$ of weighted width $O(\sqrt k)$. 
Note that a vertex of $G'$ is either a vertex of $M$, a vertex corresponding to a grid cell containing a vertex of $V_\irr$, or a vertex corresponding to a crossing point which was introduced to make the graph planar. 
For each bag of $\mathcal B'$, we replace a vertex corresponding to a grid cell $\Box$ into 
$\Box\cap V_\irr$. Also, we replace a vertex corresponding to a crossing point between two edges $e$ and $e'$ into the vertices contained in the grid cells containing the endpoints of $e$ and $e'$. 
In this way, a vertex $v$ of a bag of $\mathcal B'$ is replaced with $O(1)$ cliques. 
Moreover, the cliques are contained in the 10-neighboring cells of $\Box_v$, and thus
the clique-weight of each bag is $O(\sqrt k)$. 
Let $\mathcal B$ be the collection of the resulting bags. 

\begin{lemma}
    $(T',\mathcal B)$ is a tree decomposition of $G$. 
\end{lemma}
\begin{proof}
    First, we show that every vertex $v$ of $G$ is contained in a bag of $\mathcal B$. 
    If $v$ is a vertex of $M$, it is also contained in $G'$.
    Then it is contained in $\mathcal B$ by construction.
    If $v$ is a vertex of $V_\irr$, a vertex of $G'$ 
    corresponds to $\Box_v$. Then this vertex is contained in a bag of $\mathcal B'$, and it is replaced with the vertices of $G$ contained in $\Box_v$. Therefore, the bag contains $v$. 

    Second, we show that for every edge $uv$ of $G$, both $u$ and $v$ is contained in the same bag of $\mathcal B$. Again, if both vertices are contained in $V(M)$, then we are done. If both vertices are contained in $V_\irr$, then the are two vertices
    $u'$ and $v'$ of $G'$ corresponding to $\Box_u$ and $\Box_v$, respectively. 
    If $u'v'$ is an edge of $G'$, they are contained in the same bag of $\mathcal B'$.
    Then by construction, $u'$ and $v'$ are replaced with cliques containing $u$ and $v$, respectively. 
    If $u'v'$ is not an edge of $G'$, there is a crossing vertex $x$ lying on $u'v'$ in $G'$. Then $x$ is replaced with cliques containing $u$ and $v$ in every bag of $\mathcal B'$ containing $x$.
    The remaining case that $u'\in V(M)$ and $v' \in V_\irr$ can be handled analogously. 

    Third, we show that for any vertex $v$ of $G$,
    all bags of $\mathcal B$ containing $v$ are connected in $T'$. 
    If $v$ is a vertex of $M$, then we are done. 
    Thus we assume that $v$ is a vertex of $V_\irr$. The bags of $\mathcal B$
    containing $v$ are exactly the bags of $\mathcal B'$ containing the vertex $v'$
    corresponding to $\Box_v$ and the crossing points lying on the edges incident to $v'$. Those bags are connected in $T'$. Therefore, $(T',\mathcal B)$ is a tree decomposition of $G$. 
\end{proof}

De Berg et al.~\cite{de2020framework} showed that once we have a tree decomposition of a graph of clique-based weighted width $w$, a minimum feedback vertex set can be computed in $2^{O(\omega)}n^{O(1)}$ time, where $n$ denotes the number of vertices of the graph.
In our case, $\omega=O(\sqrt k)$ and $n=O(k)$. Therefore, we have the following theorem. 

\begin{theorem}\label{thm:fvs-query}
    Given $M$ and $V_\irr$, we can compute a feedback vertex set of $\ud(V)$ of size $k$ in $2^{O(\sqrt{k})}$ time if it exists.
\end{theorem}

\subsection{Update Algorithm}\label{sec:fvs_update}
In this subsection, we illustrate how to update data structures $V_\irr$, $M$, and $\mathcal T$ with respect to vertex insertions and deletions. 
We call $\ud(V\setminus V_\irr)$ the \emph{shell} of $\ud(V)$. 
Here, we slightly abuse the notion of the skeleton so that we can define additional \emph{special vertices} other than the boundary vertices. 
Imagine that several vertices of $V$ are predetermined as \emph{special} vertices.
The other vertices are called \emph{ordinary} vertices.
The skeleton $M$ of the shell of $\ud(V)$ with predetermined special vertices 
is defined as the minor of $\ud(V\setminus V_\irr)$ obtained by removing all degree-1 ordinary
vertices of $\ud(V\setminus V_\irr)$ repeatedly and then
contracting each maximal induced path consisting of ordinary vertices. 
Notice that if only the boundary vertices of $\ud(V)$ are set to the special vertices,
the two definitions of the skeleton coincide. 

Given an update of $V$, we first update $V_\irr$ accordingly using the update algorithm in Section~\ref{sec:hittings}. 
Recall that the number of vertices newly added to $V_\irr$ or removed from $V_\irr$ is $O(1)$.
We are to add the vertices removed from $V_\irr$ to the shell of $\ud(V)$, and remove the vertices newly added to $V_\irr$ from the shell of $\ud(V)$.
Additionally, $O(1)$ vertices of $V\setminus V_\irr$ become boundary vertices (when we handle the deletion operation), and $O(1)$ vertices of $V\setminus V_\irr$ become non-boundary vertices (when we handle the insertion operation.)
These are the only changes in the shell of $\ud(V)$ due to the update of $V$.
Note that we just need to add $v$ to the shell of $\ud(V)$ if $v$ is not in $V_\irr$.
Let $S$ be the shell of $\ud(V)$ before the update. Let $M$ be the skeleton of $S$ we currently maintain, and let $\mathcal T$ be the contracted forest concerning $M$ we currently maintain. 

In the following, we show how to update $M$ and $\mathcal T$ for the change of $S$. 
The update algorithm consists of two steps: \emph{push-pop step} and \emph{cleaning step}.
In the push-pop step, we set several vertices of $S$ as special vertices.
Specifically, the new vertices to be added to $S$ are set as special vertices.
The neighbors in $S$ of the vertices to be added to $S$ or to be removed from $S$ are set to special vertices. 
Then we update $M$ and $\mathcal T$ to the skeleton of the new set $S$ and the contracted forest concerning the new skeleton, respectively.
In the cleaning step, we make the non-boundary vertices ordinary vertices. 
Note that this changes the structure of the skeleton as well, and thus we need to update $M$ and $\mathcal T$.

\subsubsection{Push-Pop Step}
Recall that $S$ is the shell of $\ud(V)$ before the update.
Notice that the complexity of the new shell of $\ud(V)$ can be $\Theta(n)$ in the worst case.
However, the update algorithm in Section~\ref{sec:hittings} determines 
the vertices and edges added to $S$ and removed from $S$ to obtain the new shell. 
By construction, the number of such vertices and edges is $O(1)$. 
Moreover, we can determine the vertices set to the special vertices of $S$ in $O(1)$ time by Observation~\ref{obs:neighbors}.
To update $M$ and $\mathcal T$, we add (or remove) the special vertices and their incident edges one by one to (or from) $S$ until $S$ becomes the desired set.
Precisely, we add each special vertex $v$ using the \emph{push} subroutine, which adds $v$ into $M$ and updates $M$ and $\mathcal T$ accordingly. In the push subroutine,
we assume that $v$ was not contained in $S$. 
Recall that $M$ is the skeleton of $S$ obtained by removing and contracting some ordinary vertices, not special vertices.
And, we remove $v$ using the \emph{pop} subroutine, which removes $v$ from $S$ and update $M$ and $\mathcal T$.
Note that some special vertices $v$ are already in $S$.
In this case, we cannot make $v$ special using the push subroutine as it assumes that $v$ is not contained in $S$. 
For such special vertices $v$, we first pop $v$  from $S$ and then push it back to $S$. 
In this way, we can obtain the skeleton of the new shell and
the contracted forest correctly.

\begin{figure}
    \centering
    \includegraphics[width=0.9\textwidth]{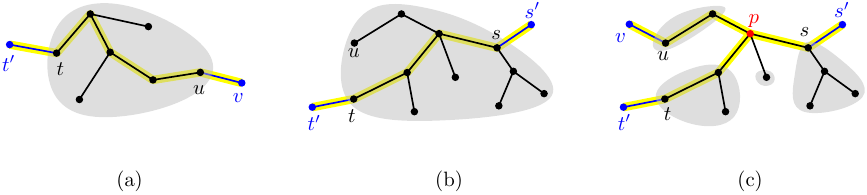}
    \caption{ (a) Case 2 of the push subroutine. After the insertion of $v$, there is a maximal path from $t'$ to $v$. The path is highlighted in yellow. (b) Case 3 of the push subroutine before the insertion of $v$. (c) Case 3 of the push subroutine after the insertion of $v$. The vertex $p$ is the lowest common ancestor of $t$ and $s$. The vertex $p$ and the edges $pv$, $pt'$, and $ps'$ are inserted into $M$.
    }
    \label{fig:push_subroutine}
\end{figure}

\subparagraph*{Push subroutine.}
Given the current shell $S$, we are to add $v$ and its incident edges to $S$. 
Let $E_v$ be the set of edges incident to $v$ to be added to $S$. 
We can ensure that the other endpoints of the edges of $E_v$ are in $S$. 
We first add $v$ and the edges of $E_v$ incident to vertices of $M$ into $S$. 
At this moment, it suffices to add $v$ and these edges to $M$. We do not need to update $\mathcal T$.
After that, we insert the remaining edges of $E_v$ into $S$ one by one as follows.
Note that those edges are incident to vertices of $\mathcal T$. 
We describe how to insert an edge $uv$ of $E_v$ with $u\in \mathcal T$ into $S$.
Let $T$ be the tree in $\mathcal T$ containing $u$.
Before the insertion of $uv$, $T$ has its bridge set.
There are the following three cases.

\subparagraph*{Case 1. $T$ has no bridge.} In this case,
$T$ was an isolated tree in $S$, and it becomes a tree connected to $M$ in $S$ due to the insertion of $uv$ into $S$. Thus it suffices to update the bridge set of $T$ to $\{uv\}$ in $O(1)$ time, and then we are done. 
\subparagraph*{Case 2. $T$ has exactly one bridge.} Let $tt'$ be the bridge of $T$ with $t\in V(T)$.
In this case, during the construction of the skeleton of $S$, all the vertices of $T$ were removed from the leaves to $t$ in order. 
However, if $uv$ were added to $S$, then $u$ would not have been removed (as $v$ is treated as a special vertex and not removed.)
As a result, none of the vertices in the $u$-$t$ path in $T$ would have been removed.
See Figure~\ref{fig:push_subroutine}(a).
Instead, the $u$-$t$ path in $T$ would have been contracted in the construction of $M$. 
To address this change in a single-shot manner, it suffices to add $vt'$ to $M$, and to add $uv$ and $tt'$ to the bridge set of $T$. This takes $O(1)$ time in total. 
\subparagraph*{Case 3. $T$ has exactly two bridges.} Let $tt'$ and $ss'$ be the bridges with $t, s\in V(T)$. 
This means that all vertices in $T$ excluding the vertices in the $t$-$s$ path 
were removed, and then all edges in the $t$-$s$ path were contracted in the construction of the skeleton.
See Figure~\ref{fig:push_subroutine}(b).
If $uv$ were added to $S$, 
no vertices in the $u$-$s$ path and the $u$-$t$ path in $T$ would have been removed during the construction of the skeleton of $S\cup \{uv\}$. 
After removing degree-1 vertices, we would be left with the vertices in the $u$-$t$ path, $u$-$s$ path and $s$-$t$ path. 
As $T$ is a tree, there must be a vertex contained in all three paths, say $p$. Then the $s'$-$p$ path, $t'$-$p$ path, and $v$-$p$ path become maximal induced paths in the resulting graph, and thus each of them would have been contracted into a single edge.
Observe that $p$ is the lowest common ancestor of $t$ and $s$ in the tree $T$ rooted at $u$.
See Figure~\ref{fig:push_subroutine}(c).
Thus we can compute $p$ using \textsf{Evert}$(\cdot)$ and \textsf{LCA}$(\cdot)$ operations supported by the link-cut trees.

To address this change in a single-shot manner, we remove $t's'$ from $M$, and add $p$, $s'p$, $t'p$, and $vp$ to $M$. Then we are done with $M$. Additionally, we update $\mathcal T$ as follows.
The deletion of $p$ from $T$ decomposes  $T$ into a constant number of subtrees each incident to at most one bridge.
Precisely, we decompose $T$ into the child subtrees of $p$, using \textsf{Cut}$(\cdot)$ operation.
For a child subtree $T'$ of $T$, we insert the edge between $p$ and the root node of $T'$ into the bridge set of $T'$.
Additionally, if $T'$ is incident to $tu$, $t'u'$, or $uv$, then we insert the corresponding edge into the bridge set of $T'$. Note that this can be checked using \textsf{Connected}$(\cdot)$ operation. In this way, we can update $\mathcal T$ and $M$ by applying a constant number of operations supported by the link-cut trees in $O(\log |V|)$ time in total.

\subparagraph*{Pop subroutine.}
We are to delete a vertex $v$ and all of its incident edges in $S$ from $S$. 
We consider two cases separately: $v$ is in $M$ or $\mathcal T$. 
For the case that $v$ is in $M$, we simply remove it and its incident edges from $M$ and the bridge sets of $\mathcal T$.
Specifically, for each edge $e$ incident to $v$ in $S$, we remove it from $M$ if it is in $M$. 
If it is not in $M$, it is in a bridge set of a tree of $\mathcal T$. 
We remove it from every bridge set.
If a bridge set containing $e$ has another edge, then there is a contracted edge of $M$, and thus we remove it.
Then we are done.  

Now consider the case that $v$ is in a tree $T$ of $\mathcal T$. 
In this case, the deletion of $v$ from $S$ changes $M$ if and only if
$T$ has two bridges $tt'$, $ss'$ with $t, s\in V(T)$ and $t',s'\in V(M)$ such that the $t$-$s$ path in $T$ contains $v$.
In particular, $M$ has the edge $t's'$, and $t's'$ must be removed from $M$ by the deletion of $v$ from $S$.
See Figure~\ref{fig:pop_subroutine}. 
We can check if the $t$-$s$ path in $T$ contains $v$ by applying $\textsf{Evert}(v)$ and $\textsf{LCA}(t,s)$. 
Then we are to remove $v$ and it incident edges from $T$ and the bridge set of $T$. Observe that an edge incident to $v$ is in $T$ or the bridge set of $T$.
We rotate $T$ in a way that $v$ becomes the root of $T$ by using $\textsf{Evert}(v)$, and remove $v$ from $T$ by applying $\textsf{Cut}(\cdot)$ operations.
Then we are given a constant number of child subtrees of $v$ since $v$ is in $V\setminus V_\irr$.
For a child subtree $T'$ of $T$ at $v$, we insert the bridges of $T$ incident to $T'$ into the bridge set of $T'$. 
Note that, given $T'$ and a bridge of $T$, we can check if $T'$ contains a vertex incident to the bridge using $\textsf{Connected}(\cdot)$.
In this way, we can remove $v$ from $\mathcal T$ and the bridge sets, and we can update $M$ accordingly in $O(\log |V|)$ time.

\begin{figure}
    \centering
    \includegraphics[width=0.45\textwidth]{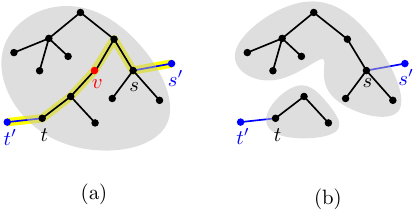}
    \caption{(a) Illustration of pop subroutine in the case that the $t$-$s$ path contains $v$. (b) Illustration of the tree $T$ after $v$ is deleted. }
    \label{fig:pop_subroutine}
\end{figure}

\subsubsection{Cleaning Step}
So far, we have treated all vertices that can cause some changes due to the update of $V$ and special vertices, and we have computed the skeleton of $\ud(V\setminus V_\irr)$.
However, some special vertices should not be considered as special vertices if they are not boundary vertices.
To handle this, we need a \emph{cleaning process} to maintain the degree of every vertex in $M$ is at least three except the boundary vertices.

We handle the vertices of $M$ of degree at most two one by one and set each of them as an ordinary vertex if it is a special vertex, as follows.
Let $v$ be a vertex we are to handle. 
First, we can check in $O(1)$ time if $v$ is a boundary vertex using Observation~\ref{obs:neighbors}.
If $v$ is a boundary vertex, we are done.
Otherwise, it suffices to handle the case that the degree of $v$ in $M$ is less than three. 
The update of $M$ is simple: 
we remove $v$ from $M$ if its degree is one in $M$, and contract a maximal induced path containing $v$ in $M$ if the degree of $v$ is two in $M$. 
Then we need to update $\mathcal T$ accordingly as follows. 
It suffices to  
merge all trees in $\mathcal T$ incident to $v$ together with their bridges incident to $v$ 
using $\textsf{Link}(\cdot)$ sequentially.
Then the bridge set of the resulting tree is the union of all bridges of merged trees except the ones incident to $v$. We can handle $v$ in $O(\log |V|)$ time. 
However, this process might decrease the degree of some other vertex of $M$ of degree in $M$ to less than three. For each such vertex, we remove it or contract a maximal induced path containing it as we did for $v$. 

\begin{lemma} \label{lem:cleaning_step_time}
    The cleaning step takes $O(\log |V|)$ time in total. 
\end{lemma}
\begin{proof}
Since we introduced $O(1)$ special vertices of $S$ and the degree of every vertex of $M$ has a constant degree, there are $O(1)$ vertices of degree at most two after the push-pop step.
As observed before, we can handle a vertex $v$ itself in $O(\log |V|)$ time.
However, the deletion of $v$ decreases the degree of some other ordinary vertices of $M$ to less than three.
Thus, we need to analyze the total number of such vertices.

Recall that we remove a vertex in $M$ if its degree is at most one, otherwise, we contract or leave it.
Furthermore, the contraction of a vertex in $M$ does not affect the degrees of any other vertices in $M$.
Thus a vertex that has a degree of at least three after the push-pop step has a degree of at least two after the deletion of $v$ from $M$.
Then we contract or leave such a vertex, and it does not affect the degrees of any other vertices in $M$.
Recall that the degree of $v$ is at most two, and it implies that the number of vertices of which the degree is decreased by the deletion of $v$ is $O(1)$

Therefore, the number of vertices handled in the cleaning step is $O(1)$, and each vertex can be handled in $O(\log |V|)$ time as we use $O(1)$ operations supported by the link-cut tree data structure.
\end{proof}

Since both steps can be done in $O(\log |V|)$ time, we have the following lemma.
\begin{lemma}
    Given a vertex update of $V$, we can update $V_\irr$, $M$ and $\mathcal T$ in $O(\log |V|)$ time.
\end{lemma}

\begin{theorem}\label{thm:ths_perf}
    There is an $O(n)$-sized fully dynamic data structure on the unit disk graph induced by a vertex set $V$ supporting $O(\log |V|)$ update time that allows us to compute a feedback vertex set of size at most $k$ in $2^{O(\sqrt k)}$ time. 
\end{theorem}

\section{Dynamic Cycle Packing Problem}\label{sec:cpp}
In this section, we describe a fully dynamic data structure on the unit disk graph of a vertex set $V$ dynamically changing under vertex insertions and deletions that can answer cycle packing queries efficiently. Each query is given with a positive integer $k$ and asks to return a set of $k$ vertex-disjoint cycles of $\ud(V)$ if it exists.
Here, we use the core grid cluster $\boxplus_\irr$, the set $V_\irr$ of vertices contained in $\boxplus_\irr$, the skeleton $M$ of $\ud(V\setminus V_\irr)$, and the contracted forest $\mathcal T$ along with the bridge sets.
They can be maintained in $O(\log |V|)$ time as shown in Section~\ref{sec:fvs_update}.
Note that $\mathcal T$ and the bridges are the auxiliary data structures for updating $M$ efficiently.
In this section, we present a query algorithm for \textsc{Cycle Packing} assuming that we have $\boxplus_\irr$, $V_\irr$, and $M$.
The following lemma was given by An and Oh~\cite{an2024eth}.

\begin{lemma}\label{lem:cp-kernel-bound}
    If $(\ud(V), k)$ is a \textbf{no}-instance for \cpp, then $|V(M)|=O(k)$.
\end{lemma}
\begin{proof}
    Observe that the number of boundary vertices of $M$ is $O(k)$.
    The other vertices of $M$ have degree at least three.
    Let $\Delta$ be the maximum degree of $M$, which is $O(1)$ by construction. 
    Now it is sufficient to show that the number of vertices of $M$ of degree at least three is $O(k)$. 
    For the purpose of analysis, we iteratively remove degree-1 vertices from $M$. Let $M_2$ be the resulting graph. In this way, a vertex of degree at least three in $M$ can be removed, and the degree of a vertex can decrease due to the removed vertices.
    We claim that the number of such vertices is $O(k)$ in total. 
    To see this, consider the subgraph of $M$ induced by the vertices removed during the construction of $M_2$. They form a rooted forest 
    such that the root of each tree is adjacent to a vertex of $M_2$.
    The number of vertices of the forest of degree at least three is linear in the number of leaf nodes of the forest. 
    By construction, 
    all leaf nodes of the forest are boundary vertices. 
    Therefore, only $O(k)$ vertices of degree at least three in $M$ are removed during the construction of $M_2$. 
    Also, the degree of a vertex decreases only when it is adjacent in $M$ to a root node of the forest.
    Since the number of root nodes of the forest is $O(k)$, and a root node of the forest is adjacent to only one vertex of $M_2$, there are $O(k)$ vertices of $M_2$ whose degrees decrease during the construction of $M_2$.

    Therefore, it is sufficient to show that vertices of $M_2$ of degree at least three is at most $12\Delta k$. 
    For this, we show that there are $|F_2|/(6\Delta)$ disjoint
    cycles in $M_2$, where  $F_2$ denotes the number of faces of $M_2$. 
    Due to the following claim, we can color the faces of $M_2$
    using $6\Delta$ colors such that the faces colored with the same
    color are pairwise vertex-disjoint.
    Thus there are $|F_2|/(6\Delta)$ disjoint
    cycles in $M_2$. Therefore, if $|F_2|$ exceeds $6\Delta k$, $(G,k)$ is a \textbf{yes}-instance. 

    \begin{claim}
        For any plane graph $H$ with maximum degree $\Delta$, the faces of $H$ can be colored with $6\Delta$ colors such that two faces sharing a common vertex are colored with different colors. 
    \end{claim}
    \begin{proof}
        We prove the claim using the induction on the number of faces of $H$.
        In the induction step, we remove edges one by one. Note that
        the removal of an edge does not decrease the maximum degree of a graph. 
        If the number of faces is less than $6\Delta$, then 
        we are done. Otherwise, we remove all degree-1 vertices iteratively. 
        Then, for each degree-2 vertex, we connect its two neighbors by an edge and remove it. Here, we do not change the drawing of $H$. In particular, there is a one-to-one correspondence between faces before and after the removal. Moreover, if two faces share a common vertex before the removal, then they still share a common vertex after the removal, and vice versa. 
        
        Consider the \emph{face-vertex graph} $H^*$ of $H$: For each face $f$ of $H$, we add a vertex for $H^*$ 
        in the interior of $f$, and connect it to all vertices incident to $f$. All vertices and edges of $H$ are also added to $H^*$. 
        Since $H^*$ is planar, it has a vertex $v$ of degree at most five.

        If $v$ corresponds to a face $f$ of $H$, then $f$ consists of at most five vertices. Therefore, there are at most $5\Delta$
        faces sharing a vertex with $f$. We remove an arbitrary edge $e$ incident to $f$, then two faces, say $f$ and $f'$, incident to $e$ are merged into a face, say $f''$. We apply the induction hypothesis to the resulting graph. 
        Then we color $f'$ using the color for $f''$.
        In this way, all faces of $H$, except for $f$, are colored properly. 
        Since the number of faces sharing a vertex with $f$ is $5\Delta$, there
        is at least one available color for $f$. Therefore, the induction step works.

        If $v$ is a vertex of $H$, it is incident to at least three faces since 
        we removed all degree-1 and degree-2 vertices of $H$. 
        This implies that $v$ is adjacent to at least six vertices in $H^*$, three for vertices of $H$ and three for vertices corresponding to faces of $H$. This contradicts the choice of $v$. Therefore, this case does not happen, and the claim holds.
    \end{proof}

    Now consider the case that $(G,k)$ is a \textbf{no}-instance. 
    Let $V_2, E_2$ and $F_2$ be the vertex set, edge set, and face set of $M_2$, respectively.
    Let $\ell$ be the number of vertices of $M_2$ of degree at least three. We have three formulas. Here $\textsf{deg}(v)$ denotes the degree of $v$ in $H_2$. 
    \begin{align*}
        &2(|V_2|-\ell)+3\ell \leq \sum_{v\in V_2} \textsf{deg}(v) = 2|E_2|\\
        &|V_2|-|E_2|+|F_2|=2 \\
        &|F_2|\leq 6\Delta k
    \end{align*}
    This implies that $\ell \leq 12\Delta k$, and thus the lemma holds. 
\end{proof}

As in Section~\ref{sec:fvs-kernel}, we first compute the graph $G$ 
by gluing $\ud(V_\irr)$ and $M$. More precisely, 
the vertex set of $G$ is the union of $V_\irr$ and $V(M)$.
There is an edge $uv$ in $G$ if and only if $uv$ is either an edge of $\ud(V_\irr)$, an edge of $M$, or an edge of $\ud(V)$ between $u\in V_\irr$ and $v\in V(M)$.
Notice that $v$ is a boundary vertex of $M$ in the third case.
By the following lemma, we can compute $k$ vertex-disjoint cycles of $\ud (V)$ by computing $k$ vertex-disjoint cycles of $G$.

\begin{lemma}\label{lem:cp-instance}
    $(G,k)$ is a \textbf{yes}-instance of $\cpp$ if and only if $(\ud(V),k)$ is a \textbf{yes}-instance of \cpp.
\end{lemma}
\begin{proof}
    Let $V_1$ be the union of $V_\irr$ and the set of boundary vertices of $\ud(V\setminus V_\irr)$. Note that $V(G)$ contains all vertices of $V_1$.
    We denote the complement $V\setminus V_1$ by $V_2$.
    For the forward direction, 
    let $\Gamma$ be the set of $k$ vertex disjoint cycles in $G$.
    For a cycle $\lambda\in \Gamma$, we compute a cycle $\lambda'$ in $\ud(V)$ as follows. For each edge $uv$  of $\lambda$, either there is an edge $uv$ on $G$ or there is a $u$-$v$ path on $\ud(V\setminus V_\irr)$. 
    For the latter case, we replace $uv$ with the $u$-$v$ path in $\ud(V\setminus V_\irr)$. We denote the resulting cycle by $\lambda'$. 
    We claim that two cycles $\lambda_1'$ and $\lambda_2'$ from $\lambda_1, \lambda_2\in \Gamma$ are vertex-disjoint.
    Since $\lambda_1, \lambda_2$ are vertex-disjoint in $G$, they are edge-disjoint. Moreover, for each vertex $w\in V_2\setminus V(G)$, there is at most one edge $uv$ in $V(G)$ such that there is a $u$-$w$-$v$ path in $\ud(V)$. Thus, $w$ is contained in at most one of $\lambda_1'$ and $\lambda_2'$. Hence, $\lambda_1'$ and $\lambda_2'$ are vertex-disjoint.
    
    For the converse, let $\Gamma'$ be the set of $k$ vertex disjoint cycles in $\ud(V)$. 
    For a cycle $\lambda'\in \Gamma'$, we do the converse way to compute a cycle $\lambda$ in $G$. In particular, for each maximal path of $\lambda'$ with all internal vertices contained in $V_2$ and have degree two in $\ud(V)$, let $u,v$ be the two end vertices of the path. Then we replace the path with $uv$ edge in $G$. The resulting graph is a cycle in $G$. Since $\Gamma'$ consists of vertex-disjoint cycle in $\ud(V)$ and $V(\lambda)\subset V(\lambda')$ holds for all cycles $\lambda'$, we obtain $k$ vertex-disjoint cycles in $G$.
    This completes the proof.    
\end{proof}

\subsection{Surface-Cut Decomposition of $G$}\label{sec:cp-surfacecut}
An and Oh~\cite{an2024eth} presented the static algorithm for \cp on unit disk graphs based on the similarity between unit disk graphs and planar graphs.
They first computes a weighted planar graph $H$, called the \emph{map sparsifier}, from the given unit disk graphs with $O(k)$ vertices together with its geometric representation. Then they recursively decompose the plane with respect to $H$.
Lastly, they design a dynamic programming algorithm for \cp on the recursive decomposition.
The efficiency of their algorithm is supported by the following two arguments. 
First, the total clique-weights of the vertices of $H$ lying on the boundary of the decomposition is \emph{small}. Second, the number of subproblems in the dynamic programming is bounded by a single exponential of the total clique-weights.

In our case, $G$ is not necessarily a unit disk graph, and thus we cannot use
the algorithm of~\cite{an2024eth} as a black box. But we can apply their argument to $G$ similarly. In particular, once we can define a map sparsifier of $G$ in a similar manner, 
we can recursively decompose the plane with respect to $H$ as they did. This
decomposition is called the \emph{surface-cut decomposition} of the plane with respect to $H$. 
Then using the surface-cut decomposition, we can design an efficient dynamic programming algorithm of \cpp.

\begin{figure}
    \centering
    \includegraphics[width=0.5\textwidth]{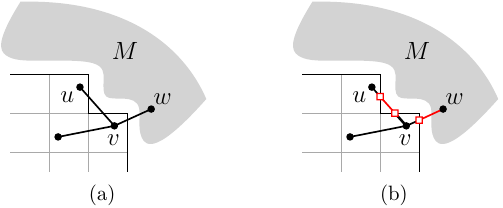}
    \caption{(a) Illustration of a part of $G$. The black grid line is a part of the boundary of $\boxplus_\irr$, and $uv$ and $wv$ are boundary edges of $G$. (b) The red boxes are the crossing points, and they split boundary edges accordingly. Then, the red edges are the subedges of the boundary edges lying outside of $\boxplus_\irr$, and the closure of $M$ contains the red edges. Note that $uv$ is a boundary edge split into three subedges.}
    \label{fig:closure-M}
\end{figure}

\subparagraph*{Drawing of $G$.}
Since $G$ is a minor of $\ud(V)$, we can draw $G$ based on $\ud(V)$.
But this drawing might have complexity $\Theta(n^2)$ in the worst case.
Instead, we consider a drawing of $G$ of complexity $\text{poly}(k)$ as follows.
Since the number of vertices of $V_\irr$ is $O(k)$, it suffices to draw $M$ and the edges between $V(M)$ and $V_\irr$ properly. 
Here, we call an edge of $G$ whose one endpoint is contained in $V_\irr$ and crosses the boundary of $\boxplus_\irr$ a \emph{boundary edge} of $G$. 
For each boundary edge $uv$ of $G$, we split it into (at most three) subedges with respect to the crossing points between $uv$ and the boundary of $\boxplus_\irr$. 
The graph obtained from $M$ by adding all subedges of the boundary edges lying outside of $\boxplus_\irr$ is called the \emph{closure} of $M$. 
See Figure~\ref{fig:closure-M}.

Consider a planar drawing of the closure of $M$ into the plane such that all 
vertices of the closure of $M$  
are drawn on their corresponding points in $V$, and the drawing does not cross 
the boundary of $\boxplus_\irr$. There always exists such a drawing. To see this, imagine that we draw the closure of $M$ in the plane based on the drawing of $\ud(V)$. 
If an edge in this drawing of the closure of $M$ crosses the boundary of $\boxplus_\irr$, we slightly perturb 
this edge so that the drawing does not cross the boundary of $\boxplus_\irr$. 
In Section~\ref{sec:planar_drawing}, we will see that we can compute a drawing of the closure of $M$ of complexity $\text{poly}(k)$ satisfying the desired properties in $\text{poly}(k)$ time.
Since the crossing points between the boundary of $\boxplus_\irr$ and 
the boundary edges of $G$ are added to the closure of $M$, 
we can draw $G$ by merging this drawing with the sub-drawing of $\ud(V)$
restricted to the edges and vertices of $G$ not appearing in $M$. 
Since the number of vertices and edges not appearing in $M$ is $O(k)$, we can compute the resulting drawing in $\text{poly}(k)$ time in total.

\begin{figure}
    \centering
    \includegraphics[width=0.7\textwidth]{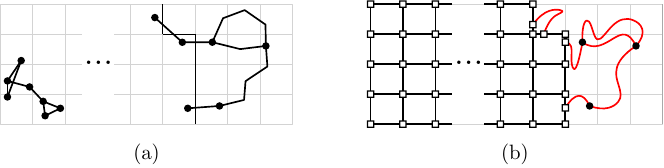}
    \caption{(a) Illustration of a part of the drawing $G$ respecting $\ud(V)$. The black grid line is a part of the boundary of $\boxplus_\irr$. (b) The map sparsifier of $G$. The black line segments and black boxes are the grid graph of $\boxplus_\irr$, and the red curves are the drawing of the closure of $M$ of complexity $\text{poly}(k)$.}
    \label{fig:map-sparsifier}
\end{figure}

\subparagraph*{Map sparsifier of $G$.}
We define the map sparsifier $H$ of $G$ and its planar drawing as follows. 
We construct the grid graph induced by the boundaries of the cells of $\boxplus_\irr$. 
In particular, each vertex (and edge) of the grid graph corresponds to a corner (and side) of a cell of $\boxplus_\irr$.
To obtain $H$, we glue the closure of $M$ and the grid graph together by considering the boundary edges $uv$ of $G$.
See Figure~\ref{fig:map-sparsifier}. 
Here, we add the vertices of the closure of $M$ added on the boundary of $\boxplus_\irr$ to $H$, and split the edges from the grid graph accordingly. 
Once we have a proper drawing of the closure of $M$ of complexity $\text{poly}(k)$, we can compute $H$ along with its geometric drawing in $\text{poly}(k)$ time.
We define the \emph{clique-weight} of each vertex of $H$ as follows. 
Recall that a vertex $v$ of $H$ is either a vertex of $M$, a vertex of the grid graph, or a newly created vertex due to the boundary edges. 
In any case, we specified the location of $v$ in the drawing of $H$. 
The clique-weight $c_H(v)$ of $v\in V(H)$ is defined as $1+\sum\log(1+|V(G)\cap \Box|)$ where the summation is taken over all $10$-neighboring cells $\Box$ of $\Box_v$. 
Let
$\wmax{G}=\max_{v\in V} c_H(v)$ and $\wsqr{G}=\sqrt{\sum_{v\in V}(c_H(v))^2}$.
\begin{lemma}\label{lem:cp-weight-of-H}
    If $|V_\irr|+|V(M)|=O(k)$, 
        {$\wmax{H} = O(\log k)$} and 
        $\wsqr{H} = O(\sqrt{k})$.
\end{lemma}
\begin{proof}
    In the definition of $c_H(v)$, we take summation of at most $O(1)$ terms $\log(1+|V\cap \Box|)$, each of which is at most $\log(1+|V_\irr|)=O(\log k)$. Hence, $\wmax{H}$ is clearly $O(\log k)$.
    For the second statement, the square of $1+\sum\log(1+|V\cap \Box|)$ is at most $O(1)\cdot (1+\sum(\log(1+|V\cap \Box|))^2)$ by the Cauchy-Schwarz inequality.
    For each grid cell $\Box\in \boxplus_\irr$, it contributes $O(1)$ clique-weights $c_H(v)$ because there are $O(1)$ 10-neighboring cells of $\Box$, and each cell contains $O(1)$ newly created vertices of $V(H)$ on its boundary.
    Hence, the squared sum of $c_H(v)$ is at most $|V(H)|+\sum O((\log(1+|V\cap \Box|))^2)\leq O(k+\sum(\log(1+k))^2)\leq O(k)$.
    Here, the last inequality comes from $(\log(z+1))^2\leq 2z$ for every positive integer $z$. This completes the proof.
\end{proof}

\subparagraph*{Surface-cut decomposition.} 
    Given the map sparsifier $H$ and its geometric drawing, we compute the \emph{surface-cut decomposition} for $H$ introduced by An and Oh~\cite{an2024eth}. 
    A point set $A$ is said to be \emph{regular closed} if the closure of the interior of $A$ is $A$ itself. 
    We call a regular closed and interior-connected region $A$ an \emph{$s$-piece} if $\mathbb{R}^2\setminus A$ has $s$ connected components. If $s=O(1)$, we simply call $A$ a \emph{piece}. Note that a piece is a planar surface, and this is why we call the following decomposition the \emph{surface-cut} decomposition. 
    The boundary of each component of $\mathbb{R}^2\setminus A$ is a closed curve. We call such a boundary component a \emph{boundary curve} of $A$.
    A \emph{surface-cut decomposition} of a plane graph $H$ with vertex weights $c: V(H)\rightarrow \mathbb{R}^+$ is a pair $(T, \mathcal A)$ 
    where $T$ is a rooted binary tree, and 
    $\mathcal A$ is a mapping that maps a node $t$ of $T$ into a piece $A_t$ in the plane satisfying the conditions \textbf{(A1--A3)}. 
    For a node $t$ and two children $t',t''$ of $t$,
    \begin{itemize} 
        \item \textbf{(A1)} $\partial A_t$ intersects the planar drawing of $H$ only at its vertices,    
	  \item \textbf{(A2)} $A_{t'}, A_{t''}$ are interior-disjoint and $A_{t'}\cup A_{t''}=A_t$, and
        \item \textbf{(A3)} $|V(H)\cap A_t|\leq 2$ if $t$ is a leaf node of $T$.
    \end{itemize}
    The \emph{clique-weighted width} of a node $t$ is defined as the sum of the clique-weights $c_H(v)$ for all $v\in \partial A_t\cap V(H)$.
    The \emph{clique-weighted width} of a surface-cut decomposition $(T, \mathcal A)$ is defined as the 
    maximum clique-weight of the nodes of $T$.
    An and Oh~\cite{an2024eth} showed that
    a planar graph has a surface-cut decomposition of small clique-weighted width, and it can be computed in time polynomial in the complexity of the planar graph. 
\begin{theorem}\label{thm:surfacecut}
    If $|V_\irr|+|V(M)|=O(k)$ and the complexity of the drawing of $M$ is $\text{poly}(k)$, one can compute surface-cut decomposition $(T, \mathcal A)$ of clique-weighted width $O(\sqrt k)$ in $\text{poly}(k)$ time. Moreover, the number of nodes in $T$ is $O(k)$.
\end{theorem}

\subsection{Dynamic Programming Algorithm}\label{sec:cp-dp}
Once we have a drawing of $G$ and a surface-cut decomposition $(T,\mathcal A)$ of $H$ with the weighted width $\omega$,
we can solve \cp on $G$ in $2^{O(\omega)}$-time using the algorithm in~\cite{an2024eth} with a slight (straightforward) modification. 
However, to make our paper self-contained, 
we describe the high-level summary of the dynamic programming algorithm in~\cite{an2024eth}. 

\subparagraph*{Structured solutions.}
We say a cycle $C$ of $G$ is an \emph{intra}-cycle if all vertices of $C$ are contained in a single cell $\Box$  and an \emph{inter}-cycle otherwise. 
If $(G,k)$ is a \textbf{yes}-instance of \cpp, there are $k$ vertex-disjoint cycles such that the number of inter-cycles intersecting a single grid cell $\Box$ is at most $O(1)$. 
Let $\Gamma$ be the set of vertex-disjoint cycles of $G$. Since we have the drawing of $G$, we can define the drawing of the cycles of $\Gamma$ appropriately. The \emph{intersection graph} of  $\Gamma$ is defined as a graph whose vertex set is $\Gamma$ such that two vertices are connected by an edge if and only if the drawings of two corresponding cycles intersect. If $(G,k)$ is a \textbf{yes}-instance of \cpp, there are $k$ vertex-disjoint cycles whose corresponding intersection graph is $K_{z,z}$-free for some constant $z$.

\begin{lemma}[Lemma 22 and 23 of~\cite{an2024eth}] \label{lem:sol-structure}
Suppose $(G,k)$ is a \textbf{yes}-instance. Then $G$ has a set $\Gamma$ of $k$ vertex-disjoint cycles such that for each grid cell $\Box$, there are $O(1)$ inter-cycles of $\Gamma$ intersecting $\Box$, and the intersection graph of $\Gamma$ is $K_{z,z}$-free for $z=O(1)$.
\end{lemma}

From now on, we let $\Gamma$ denote the maximum number of vertex-disjoint cycles of $G$ that satisfies the condition specified in Lemma~\ref{lem:sol-structure}.
We say $\Gamma$ as a \emph{structured solution}.
If it is clear from the context, we use $\Gamma$ as a set of vertex-disjoint cycles and a subgraph of $G$ interchangeably. 

\begin{figure}
    \centering
    \includegraphics[width=0.9\textwidth]{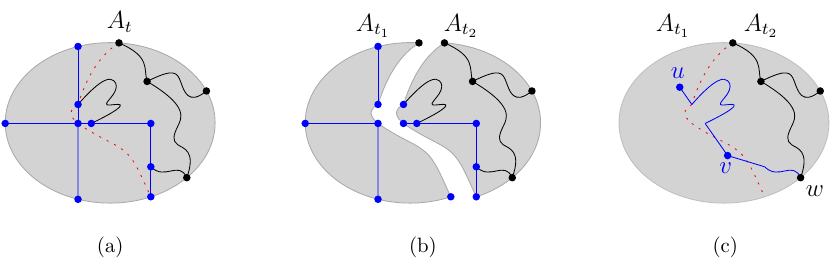}
    \caption{(a) The boundary of a piece $A_t$ intersects only vertices of $H$. The blue vertices and lines are a part of the grid graph, and the black ones are a part of the drawing of the closure of $M$. (b) In the construction of a surface-cut decomposition of $H$, $A_t$ is separated into two pieces $A_{t_1}$ and $A_{t_2}$ along the red dotted curve of (a). (c) The blue vertices and lines are the drawing of the vertices in $V_\irr$ and their incident edges. $uv$ is in $\textsf{cut}(t_1)$ and $vw$ is in $\textsf{cut}(t_1)$ and $\textsf{cut}(t_2)$. Note that every vertex of $M$ is on the boundary of pieces.}
    \label{fig:suface-cut-decom}
\end{figure}

    \subparagraph*{Subproblems.}
    For each node $t$ of $T$, we define subproblems for the subgraph $G_t$ of $G$ induced by $V(G)\cap A_t$. 
    Note that the drawing of an edge $uv$ of $G$ is no longer necessarily a line segment.
    Let $\cut t$ denote the set of edges of $G$ whose one endpoint is contained in $A_t$ and whose drawing intersects $\partial A_t$. 
    See Figure~\ref{fig:suface-cut-decom}. 
    Since a vertex of $G_t\setminus V(\cut t)$ is never incident to a vertex of $G\setminus G_t$,
    we are enough to encode the interaction between $\Gamma$ and $\cut t$.
    To do this, we focus on the parts of $\Gamma$ consisting of edges of $\cut t$.
     In particular, the interaction can be characterized as tuple $(\Lambda, V(\Delta), \mathcal P)$, where $\Lambda$ denotes the subgraph induced by edges of $\cut t\cap \Gamma$, $\Delta$ denotes the set of all intra-cycles of $\Gamma$, and $\mathcal P$ denotes the pairing of the endpoints of the paths of $\Lambda$ contained in $A_t$. Here, $\mathcal P$ stores the information that two vertices of each pair belong to a single path 
     in the subgraph of $\Gamma$ induced by $V(G)\cap A_t$. We call $(\Lambda, V(\Delta), \mathcal P)$ the \emph{signature} of $\Gamma$.

    Since we are not given $\Gamma$ in advance, we try all possible tuples $(\widetilde \Lambda, \widetilde V, \widetilde{\mathcal P})$ that can be the signatures of $\Gamma$, and find an optimal solution for each.
    Let $\boxplus_t$ be the set of grid cells that contain a vertex of $\cut t$.
    Let $f$ denote the mapping that maps an endpoint $u$ of an edge $uv$ of $\cut t$ into the first intersection point between the drawing of $uv$ and $\partial A_t$ starting from $u$. Clearly, $f$ is a one-to-one mapping.
    We impose three conditions for each node $t\in T$ and each cell $\Box$ of $\boxplus_t$. In the third condition, we let $\textsf{end}(\widetilde \Lambda)$ denote the image of $f$ on the set of endpoints of the paths of $\widetilde \Lambda$ contained in $A_t$. 
    \begin{itemize}\setlength\itemsep{-0.1em}
        \item $\widetilde \Lambda$ induces vertex-disjoint cycles and paths of $\cut t$, 
        \item $\widetilde V$ is a subset of 
        $\{v\in V(G): \Box_v\subseteq \boxplus_t\}$ with $\widetilde V\cap V(\widetilde \Lambda)=\emptyset$ and $|\Box_v \cap (V(G)\setminus \widetilde V)|=O(1)$, and
        \item $\widetilde{\mathcal P}$ is a pairing of $\textsf{end}(\widetilde \Lambda)$.
    \end{itemize}
    The issue here is that the third condition gives $2^{O(\omega\log \omega)}$ possible pairings, which violates the desired complexity. 
    To handle this issue, we strengthen the third condition based on Lemma~\ref{lem:sol-structure}.
    In particular, 
    if $(\widetilde \Lambda, \widetilde V, \widetilde{\mathcal P})$ is the signature of a \emph{structured} solution $\Gamma$, there is a set of curves contained in $A_t$ such that each curve connects two points of $\partial A_t$ specified by a pair of $\widetilde{\mathcal P}$. 
    Then the intersection graph with respect to these curves is
    $K_{z,z}$-free by Lemma~\ref{lem:sol-structure}. Subsequently, we can strengthen the third condition as follows.
    \begin{itemize}
        \item $\widetilde{\mathcal P}$ is a pairing of $\textsf{end}(\widetilde \Lambda)$ such that there is a set of curves connecting two points of the pairs of $\widetilde{\mathcal P}$ whose intersection graph is $K_{z,z}$-free for a constant $z$. 
    \end{itemize}

    \begin{lemma}\label{lem:tuple-num}
        The number of tuples $(\widetilde \Lambda, \widetilde V, \widetilde{\mathcal P})$ satisfying the posed conditions is $2^{O(\omega)}$. Moreover, these tuples can be computed in $2^{O(\omega)}$ time. 
    \end{lemma}
    \begin{proof}
        First, we analyze the number of choices of $\widetilde\Lambda$. 
        For each edge $uv$ of $\cut t$, we charge $u$ and $v$ into some vertex of $\partial A_t\cap V(H)$ through the case studies.
        In the first case, the drawing of $uv$ and $\partial A_t$ intersect outside \hgc.
        Since $\partial A_t$ never intersects the edge part of $H$ that is drawn outside \hgc, without loss of generality, $u$ is an intersection point between $uv$ and $\partial A_t$. 
        We charge $u$ and $v$ into $u$.
        In the other case, $uv$ and $\partial A_t$ intersect inside \hgc. Let $w$ be an intersection point.
        Since $H$ has no vertex on the interior of $\Box_w$, there is a vertex of $V(H)$ lie on $\partial \Box_w\cap \partial A_t$.
        We charge $u,v$ to that vertex. 
        Note that for both cases, we charge $u$ and $v$ into a vertex of $V(H)$ such that the distance between $u$ (and $v$) and the vertex of $V(H)$ is at most two. 
        Subsequently, for each $v\in V(H)\cap \partial A_t$, $u$ can be charged into $v$ only if $u$ is contained in a 10-neighboring cell of $\Box_v$.
        Hence, the number of    
        different sets of $V(\widetilde \Lambda)$ is at most
        \[ 
        \prod_{\substack{d(\Box_v, \Box)\leq 10\\ v\in \partial A_t}}|V(G)\cap\Box|^{O(1)}= 
        \exp 
        \biggl( \sum_{\substack{d(\Box_v, \Box)\leq 10\\ v\in \partial A_t}} O(\log |V(G)\cap \Box|)\biggr)
        =
        \exp\biggl( 
        O \biggl( \sum_{v\in\partial A_t} c_H(v)\biggr)\biggr) = 2^{O(\omega)}.\]
        Since each vertex of $\widetilde \Lambda
        $ has degree at most two in $\widetilde \Lambda
        $, the number of different edge sets of $\widetilde \Lambda
        $ is $2^{O(\omega)}$.
        
        Second, the number of different $\widetilde V$ is $2^{O(\omega)}$ because all but $O(1)$ vertices are contained in $\widetilde V$ for each cell $\Box$ of $\boxplus_t$, and the number of cells of $\boxplus_t$ is at most the number of vertices in $\partial A_t$, which is $O(\omega)$.
        All possible pairs $(\widetilde \Lambda, \widetilde V)$ can be computed in $2^{O(\omega)}$ time.
        An and Oh~\cite{an2024eth} showed that the number of different possible pairings $\widetilde{\mathcal P}$ under the third posed condition is $2^{O(\omega)}$, and we can enumerate all possible pairings in $2^{O(\omega)}$ time. 
    \end{proof}

    Combining all, we conclude the following theorem.

\begin{theorem}\label{thm:cp-query}
    Given the core grid cluster \hg and the skeleton $M$, \cp can be solved in $2^{O(\sqrt k)}$ time.
\end{theorem}
\begin{proof}
    Due to the Lemma~\ref{lem:cp-kernel-bound}, $(G,k)$ is \textbf{yes}-instance if either $|V_\irr|=\omega(k)$ or $V(M)=\omega(k)$. 
    Otherwise, we compute the drawing of $H$ using Lemma~\ref{lem:drawing-planar-polygonial} and compute the surface-cut decomposition $(T, \mathcal A)$ of clique-weighted width $O(\sqrt k)$ in $\text{poly}(k)$ time using Theorem~\ref{thm:surfacecut}.
    In the dynamic programming on $(T, \mathcal A)$, the number of tuples we have enough to consider is $2^{O(\sqrt k)}$ and these tuples can be computed in $2^{O(\sqrt k)}$ time. 
    Then the dynamic programming takes $2^{O(\sqrt k)}$ time since $T$ has $O(k)$ nodes.
\end{proof}

\subsection{Planar Drawing of the Closure of $M$}\label{sec:planar_drawing}
In this subsection, we illustrate how to compute a planar drawing of the closure of $M$ into the plane such that the 
vertices are drawn on their corresponding points in $V$, and the drawing does not intersect
the boundary of $\boxplus_\irr$.  
Here, it suffices to prove the following lemma. 
In our case, each vertex of the closure of $M$ is prespecified,  
and each connected region of $\mathbb{R}^2\setminus \boxplus_\irr$ is a polygonal domain $\Sigma$. 
By applying the following lemma for each connected region of $\mathbb{R}^2\setminus \boxplus_\irr$, we can compute a planar drawing of the closure of $M$ of complexity $\text{poly}(k)$ with the desired properties in $\text{poly}(k)$ time. 

\begin{lemma}\label{lem:drawing-planar-polygonial}
    Any planar graph $L$ admits a planar drawing on a polygonal domain $\Sigma$ that maps each vertex to its prespecified location and each edge to a polygonal curve with $O(|\Sigma|\cdot |E(L)|^3)$ bends, where $|\Sigma|$ denotes the total complexity of boundaries of $\Sigma$. Moreover, we can draw such one in $O(|\Sigma|\cdot |E(L)|^4)$ time if we know the proper ordering of edges incident to each vertex.
\end{lemma}

We demonstrate Lemma~\ref{lem:drawing-planar-polygonial} by Witney's theorem~\cite{whitney19332, whitney1932congruent} along with the algorithm in~\cite{pach2001embedding}.
Specifically, Witney's theorem guarantees that a \emph{3-connected} planar graph admits a topologically unique planar embedding.

For clarity, we assume that $L$ is connected, and we demonstrate how to draw a planar drawing of $L$ that maps each vertex to its prespecified location and each edge to a polygonal curve with $O(|\Sigma|\cdot |E(L)|)$ bends in $O(|\Sigma|\cdot|E(L)^2|)$ time.
When $L$ is not connected, we can draw the desired planar drawing of $L$ by drawing each component of $L$ using the aforementioned algorithm, sequentially.
Precisely, after we compute a drawing of a connected component of $L$, we add the region containing the drawing into $\Sigma$ as a hole.
We can compute such region in time $O(|\Sigma|\cdot |E(L)|^2)$ time by unifying all faces of the drawing except the outer face.
Note that the total complexity of the boundaries of such regions is $O(|\Sigma|\cdot |E(L)|^2)$.
Thus we can compute the desired planar drawing of $L$ in $O(|\Sigma|\cdot |E(L)|^4)$ time.

\begin{lemma}[Theorem 1 in~\cite{pach2001embedding}]\label{lem:fix-drawing}
    Every planar graph $L$ admits a planar drawing that maps each vertex to an arbitrarily prespecified distinct location and each edge to
    a polygonal curve with $O(|V(L)|)$ bends.
    Moreover, such a drawing can be constructed in $O(|V(L)|^2)$ time.
\end{lemma}

We first contract each hole in $\Sigma$ into a single point and compute the planar drawing of $L$ in a plane without holes using Lemma~\ref{lem:fix-drawing}.
After that, we recover the holes in $L$ at the prespecified locations.
In the recovering step, we want to ensure that the given topological ordering is maintained in the planar drawing of $L$.
To do this, we use Witney's theorem. 
However, $L$ is not necessarily 3-connected.
For this reason, we slightly modify $L$ so that it becomes 3-connected.


\begin{figure}
    \centering
    \includegraphics[width=0.5\textwidth]{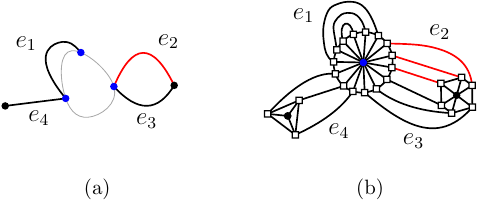}
    \caption{(a) The blue vertices denote the vertices of $L$ on the boundary of a hole. (b) The vertices on a hole are contracted into one blue vertex on the boundary of the hole. And, the black boxes denote the vertices of gadgets. Each edge is copied into three edges. For clarity, the edges corresponding to $e_2$ are colored by red.}
    \label{fig:L-tc}
\end{figure}

\subparagraph*{Modification for $L_\tc$ and drawing $\mathcal D_\tc$.}
To utilize Witney's theorem and Lemma~\ref{lem:fix-drawing}, we obtain a \emph{3-connected} planar graph $L_{\tc}$ by modifying $L$ with respect to the holes in the polygonal domain $\Sigma$, and we compute a polygonal drawing $\mathcal D_{\tc}$ of $L_{\tc}$ using 
Lemma~\ref{lem:fix-drawing}. 
We assume that each vertex in $L$ is the distance at least $0<\varepsilon<1$ from any other vertex in $L$.

We first contract each hole into a single point in the plane.
In this way, all vertices of $L$ lying on the boundary of the same hole or the outer boundary of $\Sigma$ are contracted to a single vertex on the boundary accordingly.
Next, we add a \emph{wheel graph} centered on $v$ for each vertex $v$ of $L$ as a gadget.
Specifically, the gadget is formed by 
connecting a single universal vertex to the vertices of 
a cycle with $3d_v$ vertices of diameter $\varepsilon/100$, 
where $d_v$ is the degree of $v$ in $L$.
Then we replace each edge $uv$ of $L$ with three edges:
we choose three consecutive vertices from the gadget for $u$
and three consecutive vertices from the gadget for $v$,
then we connect them.  See Figure~\ref{fig:L-tc}. 
We can do this for all edges of $L$ without crossing by maintaining the proper ordering of the edges around each vertex of $L$.
Recall that we have the proper ordering of incident edges at each vertex in $L$. 
We denote the result graph as $L_\tc$.
Note that $L_\tc$ is a 3-connected planar graph since a wheel graph with at least four vertices is 3-connected.

The number of vertices in $L_\tc$ is at most $6|E(H)|+|V(H)|$.
By Lemma~\ref{lem:fix-drawing}, we can compute a planar drawing of $L_\tc$ in the plane where each edge has at most $O(|E(H)|)$ bends.
Furthermore, it is the unique topological embedding by Witney's theorem since $L_\tc$ is a 3-connected planar graph.
We denote the drawing by $\mathcal D_\tc$.
In the following, we recover a polygonal drawing $\mathcal D$ of $L$ in  $\Sigma$.

\begin{figure}
    \centering
    \includegraphics[width=0.7\textwidth]{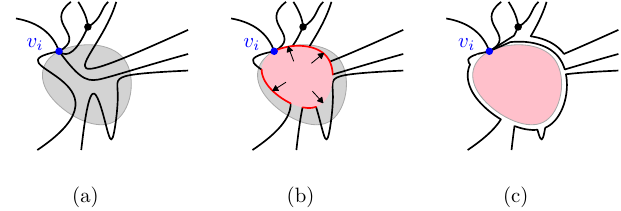}
    \caption{(a) A part of $D_\tc$ after removing paths. The gray region is a hole $A_i$, and the blue vertex is $v_i$. and (b) The pink region is $A'_i$. While blowing $A'_i$, we push subcurves on its boundary as the red curves if they intersect $A_i$. (c) After blowing up $A'_i$, we perturb the drawing of $\mathcal D$ without crossing.}
    \label{fig:Recover-D-Blow}
\end{figure}

\subparagraph*{Recovering $\mathcal D$ from $\mathcal D_\tc$.}
We recover a drawing $\mathcal D$ of $L$ in $\Sigma$ from $\mathcal D_\tc$.
For each edge $uv$ in $L$, it corresponds to three paths of length three in $L_\tc$ connecting the universal vertices of the gadgets for $u$ and $v$.
Among them, except the middle one, we remove two paths from $\mathcal D_\tc$.
While keeping the drawing of $\mathcal D_\tc$ of the remaining path between the universal vertices for $u$ and $v$ as the drawing of $uv$ of $L$, we remove the vertices in $L_\tc$ which are not in $L$.
This process increases the number of bends of each edge by a factor of at least three compared to $\mathcal D_\tc$.
We refer to the obtained drawing as $\mathcal D$.

In the following, we recover the boundaries of $\Sigma$.
Let $A_1, A_2,\dots$, and $A_\ell$ be holes of $\Sigma$,
Note that the vertices of $L$ on the boundary of $A_i$ are contracted into one vertex in $L_\tc$. 
We refer to the contracted vertex as $v_i$ for each $A_i$.


For each $A_i$, we modify $\mathcal D$ so that the resulting drawing avoids $A_i$.
Precisely, we blow up a region $A'_i$, which is initially the point $v_i$, 
within a face adjacent to $v_i$ until it becomes $A_i$. 
While blowing up $A'_i$, we push the subcurves of the curves of $\mathcal D$ which intersect $A'_i$ onto the boundary of $A'_i$.
Notice that we push such subcurves, avoiding $v_i$, except the boundary curves of the face containing $A'_i$.
See Figure~\ref{fig:Recover-D-Blow}(b).
By repeating such a process, we make the interior of $A_i$ empty.
Then, by perturbing the drawing $\mathcal D$, we can modify the drawing to avoid $A_i$ and crossing.
See Figure~\ref{fig:Recover-D-Blow}(c). 
This process increases the bends by a factor of $|\Sigma|$ compared to $\mathcal D_\tc$.


\begin{figure}
    \centering
    \includegraphics[width=0.74\textwidth]{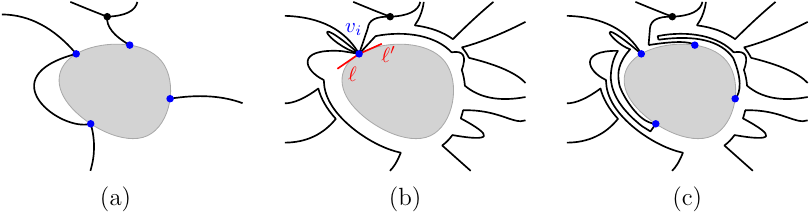}
        \caption{(a) A part of $L$. The gray region is a hole $A_i$, and the blue vertices are the vertices of $L$ on the boundary of $A_i$. (b) A part of drawing $D$ after making $A_i$ empty. The blue vertex is $v_i$, and the red lines are $l$ and $l'$, respectively. (c) We uncontract $v_i$ into the four blue vertices of $L$.
    \label{fig:Recover-D-uncontract}}
\end{figure}

In the following, we uncontract the vertices $v_i$'s for each $A_i$'s. 
Let $\ell$ and $\ell'$ be short line segments incident to $v_i$ drawn in opposite directions within the face containing $A_i$ so that they intersect $\mathcal D$ and $A_i$ only at $v_i$.
See Figure~\ref{fig:Recover-D-uncontract}(b).
Note that the edges in $\mathcal D$ incident to $v_i$ have an endpoint on the boundary of $A_i$ in $L$, and the ordering is the same as the prespecified ordering along the boundary of $A_i$ in $L$ due to Witney's theorem.
We choose the closest incident edge $e$ of $v_i$ to the boundary of $A_i$, and we uncontract an endpoint $u$ of the edge $e'$ in $L$ corresponding to $e$ contracted into $v_i$.
Precisely, when there is no incident edge of $v_i$ between $e$ and $\ell$ (or $\ell'$), we extend the drawing of $e$ along the boundary of $A_i$ in the counterclockwise direction (or clockwise direction) until $u$ is located at its prespecified location.
If both endpoints of $e'$ are contracted into $v_i$, we uncontract the endpoint preceding in the clockwise order (or counterclockwise order).
By repeating such a process, we uncontract all the vertices in $L$ on the boundary $A_i$ at the prespecified location.
See Figure~\ref{fig:Recover-D-uncontract}(c).
This process increases the bends of each edge by a factor of at most $|\Sigma|$.

\medskip

In conclusion, the obtained $\mathcal D$ is a polygonal drawing of $L$ in $\Sigma$ each of which edge has at most $O(|\Sigma|\cdot |E(H)|)$ bends, where $|\Sigma|$ denotes the complexity of the boundary of $\Sigma$.
Furthermore, the above processes take $O(|\Sigma|\cdot |E(H)|^2)$ time in total. 
This completes the proof of Lemma~\ref{lem:drawing-planar-polygonial}.

\section{Conclusion}
In this paper, we initiate the study of fundamental parameterized problems for dynamic unit disk graphs. 
including 
 \textsc{$k$-Path/Cycle}, \textsc{Vertex Cover}, \textsc{Triangle Hitting Set}, \textsc{Feedback Vertex Set}, and \textsc{Cycle Packing}. 
Our data structure supports $2^{O(\sqrt{k})}$ update time and $O(k)$ query time for \textsc{$k$-Path/Cycle}.
For the other problems, our data structures support $O(\log n)$ update time and $2^{O(\sqrt{k})}$ query time, where $k$ denotes the output size.
Despite the progress made in this work, there remain numerous open problems.
First, can we obtain a trade-off between query times and update times?
Second, one might consider other classes of geometric intersection graphs in the dynamic setting such as disk graphs~\cite{an2023faster,lokshtanov2022subexponential}, outerstring graphs~\cite{an2024sparseouterstringgraphslogarithmic}, transmission graphs~\cite{transmission} and hyperbolic unit disk graphs~\cite{oh_et_al:LIPIcs.ESA.2023.85}.
To the best of our knowledge, there have been no known results on parameterized algorithms for those graph classes.

\bibliography{papers}
\clearpage

\end{document}